\documentclass[preprint,12pt,numbers]{elsarticle}

\usepackage{graphicx}
\usepackage{amsmath}
\usepackage{amssymb}
\usepackage{amsthm}
\usepackage{booktabs}

\newtheorem{lemma}{Lemma}

\usepackage[ruled,vlined]{algorithm2e}
\usepackage{tikz}
\usetikzlibrary{positioning, fit, calc}
\usepackage{pgfplots}
\usepackage{hyperref}
\usepackage{xspace}

\pgfplotsset{compat=1.16}

\def\codename{\textsc{EnclaveScale}\xspace}

\journal{Journal of Information Security and Applications}

\begin{document}

\begin{frontmatter}

\title{\codename: Hardware-Assisted Edge-DP for Secure Data Centre Power Telemetry}

\author[1]{Hung Dang\corref{cor1}}
\ead{hung.dk@vlu.edu.vn}
\cortext[cor1]{Corresponding author}

\author[2]{Tue Nguyen}
\ead{tuent@appota.com}

\author[3]{Minh Vo}
\ead{MinhVX2@fpt.com}

\address[1]{Faculty of Information Technology, Van Lang School of Technology, Van Lang University, Ho Chi Minh City, Vietnam}
\address[2]{Appota Group, Vietnam}
\address[3]{FPT Corporation}

\begin{abstract}
\codename is a distributed, hardware-assisted telemetry architecture providing post-extraction attestation, enabling operators to collaboratively model high-resolution generative AI power transients. Existing cryptographic techniques scale poorly for 10-Hz streaming or fail to authenticate origins, permitting malicious hosts to spoof sensor inputs. We implement and evaluate a post-extraction pipeline utilizing DCAP attestation, differential privacy noise injection, and Byzantine rejection across 32 GCP Confidential VMs, achieving 0\% post-extraction attack success rate. This edge-DP approach distils continuous GPU transients into discrete Markov-chain transition matrices, guaranteeing event-level differential privacy. To mitigate pre-ingestion vulnerabilities, we propose an SPDM-authenticated first-mile layer. While current platforms lack attested I/O, emerging hardware architectures integrate PCIe IDE and TDISP to natively prevent host-level synthesis, securing the end-to-end provenance boundary. A Global Aggregation Enclave verifies these cryptographic proofs prior to capacity-weighted aggregation. Evaluation demonstrates a steady-state throughput of $131{,}406$ samples/s per enclave, amortising attestation overhead to $0.23\,\mu$s/sample. On empirical NVML-sampled H100, A100, and L4 traces, \codename achieves a dynamic orchestration margin error of $1.3$\,MW compared to $0.1$\,MW for an honest-aggregator central-DP baseline. \codename establishes a secure foundation for dynamic multi-tenant power orchestration, obfuscating sub-second anomalies locally and protecting macro-workload confidentiality via spatial dilution during global aggregation.

\end{abstract}


\end{frontmatter}

\section{Introduction}
\label{sec:introduction}

Large-scale training and inference clusters utilizing high-density accelerators (e.g., NVIDIA H100 and A100 GPUs) generate massive, synchronized power transients at the $10$-Hz scale. These transients, resulting from coordinated micro-batch executions and model-parallel communication barriers, cause instantaneous power spikes of hundreds of megawatts~\cite{VercellinoAI2024}. 

Grid provisioning requires aggregating raw telemetry across multiple independent infrastructure providers. However, high-resolution telemetry directly encodes proprietary workload schedules and microarchitectural efficiency optimizations, preventing centralized aggregation. Protecting this telemetry requires a dual-layer defense: securing individual execution anomalies at the edge and diluting macro-workload identities during global aggregation. Existing privacy-preserving paradigms fail to support high-frequency streaming telemetry securely. Cryptographic Multi-Party Computation (MPC)~\cite{KellerMPSPDZ} over per-sample shares is WAN-infeasible at 10 Hz ($O(n)$ bandwidth per sample). While MPC over condensed matrices is bandwidth-feasible, it sacrifices pre-sharing execution integrity against a malicious host that can fabricate the matrix before the MPC boundary, as we demonstrate in §\ref{sec:eval_byzantine}. Federated Learning Secure Aggregation (SecAgg)~\cite{BonawitzSecAgg} successfully reduces transit exposure but relies on honest-but-curious assumptions, offering no cryptographic defense against a malicious host spoofing the telemetry prior to encryption. Both MPC and SecAgg blindly ingest host-provided data, suffering a fundamental ``first-mile'' data provenance gap. Alternatively, centralizing raw telemetry in a Trusted Execution Environment (TEE) aggregator forces single-domain trust concentration and scales poorly over cross-region inter-DC links.

We propose \codename, a distributed, hardware-assisted framework that resolves the bandwidth bottleneck and post-extraction execution-integrity gap via attested edge sanitisation on Intel Trust Domain Extensions (TDX). In this architecture, the State-Transition Extractor (STE) computes discrete Markov-chain state-transition matrices over the stream locally, isolated from the host. The enclave applies bounded differential privacy (DP) noise to the matrix to obfuscate precise \texttt{AllReduce} barrier timings, and binds the cryptographic output to a DCAP attestation quote before WAN transmission.

\codename reduces WAN egress (§\ref{sec:eval_microbenchmarks}), mathematically verifies local sanitisation (§\ref{sec:eval_byzantine}), and bounds utility degradation to DP noise (§\ref{sec:eval_utility}). Injecting noise locally incurs a $1.3$\,MW provisioning error (a $1.2$\,MW increase over the central-DP baseline, at per-batch $\varepsilon=1$ ($\varepsilon_{\text{epoch}} \approx 11.3$ over $10$\,min, $\delta=10^{-6}$), across 32 providers) for formal execution integrity against a malicious host. Protection of the macro-workload schedule relies entirely on structural spatial dilution at the global aggregator.

\textbf{Contributions:}
\begin{itemize}
    \item We design \codename's post-extraction pipeline (DCAP attestation, DP noise injection, and Byzantine rejection) to preclude software spoofing, modeling standard TEE adversary capabilities to formally reduce execution integrity to the unforgeability of the TDX quoting key.
    \item We detail an attested TEE-local differential privacy protocol, injecting Gaussian noise calibrated to a closed-form $\ell_2$-sensitivity bound ($\Delta_2(f) = \sqrt{6}$) and deriving a tight Rényi-DP composition bound for continuous streaming release.
    \item We execute a multi-region deployment across 32 Google Cloud C3 Confidential VMs using real NVML traces (H100, A100, L4), characterising throughput, attestation overhead, scalability, and Byzantine robustness. We quantify empirical utility by reporting dynamic orchestration megawatt error against a plaintext ground truth via strict parameter sweeps, comparison with four contemporaneous baselines, and a $72$-hour stability run.
    \item We explicitly scope the pre-ingestion first-mile layer as an architectural design sketch for forthcoming Granite Rapids (Xeon 6) platforms, demonstrating how PCIe TDISP theoretically enforces end-to-end telemetry provenance against host-level manipulation.
\end{itemize}

\section{Background and Problem Formulation}
\label{sec:background}

\subsection{Background: Generative AI Power Transients}
\label{sec:bg_power}

Traditional $1$- to $5$-minute power profiling masks the profound volatility of large-scale generative AI workloads. During distributed training, clusters of high-density accelerators (e.g., NVIDIA H100) execute synchronized micro-batches, transitioning from idle ($\approx 100$\,W) to thermal design power limits ($\approx 700$\,W) within $50$ to $100$ milliseconds due to model-parallel communication barriers (e.g., \texttt{AllReduce})~\cite{VercellinoAI2024, PattersonCarbon}. At facility scale, thousands of GPUs hitting simultaneous barriers induce massive, sub-second power transients. Standard data center undervoltage protection relays operate within 1 to 3 cycles (16--50\,ms). While absolute protection against these physical trips strictly requires provisioning for the theoretical peak ($P_{\max}$), modeling the stochastic distribution of these spikes using $10$-Hz-resolution telemetry enables aggressive, dynamic power orchestration and active load shifting underneath this hard physical ceiling.

\subsection{Background: Hardware Primitives (TDX and SPDM)}
\label{sec:bg_tdx}

Intel TDX provides VM-level hardware-enforced isolation (Trust Domains) with MKTME-based memory encryption and DCAP remote attestation~\cite{IntelTDX}; GCP Confidential VMs expose these primitives to tenant VMs (§6.1). The Security Protocol and Data Model (SPDM), published by the DMTF, authenticates hardware component identity via X.509 certificate chains; SPDM 1.2+ additionally establishes secure session keys directly between hardware endpoints via \texttt{KEY\_EXCHANGE} over ECDHE (or DHE/PSK). SPDM uses an AEAD-protected session (e.g., AES-GCM or ChaCha20-Poly1305) for confidentiality and integrity, rather than relying on TLS records. This enables secure telemetry transit across an untrusted host OS and establishes hardware-rooted data provenance.

\subsection{Background: Event-Level DP for Transition Matrices}
\label{sec:bg_dp}

We achieve Differential Privacy (DP) via the Gaussian Mechanism, injecting $\mathcal{N}(0, \sigma^2)$ noise calibrated to the global $\ell_2$-sensitivity $\Delta_2(f)$. In \codename, the query $f$ maps a finite stream of discrete power states to a state-transition matrix $M$. The global $\ell_2$-sensitivity is strictly bounded at $\Delta_2(f) = \sqrt{(-2)^2 + 1^2 + 1^2} = \sqrt{6}$. Protecting a spike spanning $k$ samples requires increasing the noise scale by a factor $k$ (i.e., $\sigma \to k\sigma$) to achieve $(\varepsilon, \delta)$-DP for the $k$-sample group, or equivalently, the same noise scale $\sigma$ provides only $(k\varepsilon, k\delta)$-DP group privacy. At $k = 5$ and per-batch $\varepsilon = 1$, group privacy yields only a $(5, 5 \times 10^{-6})$-DP guarantee for $500$\,ms spikes, which provides negligible formal protection. Extensive formal DP definitions and degradation proofs are provided in the Supplementary Material.

\subsection{Threat Model}
\label{sec:threat_model}

The \codename architecture comprises four entities: \textbf{Data Providers} (infrastructure operators), \textbf{Local Sanitisation Enclaves} (LSEs, TDX Confidential VMs), a \textbf{Global Aggregation Enclave} (GAE, neutral consortium), and \textbf{Verifiers} (grid operators). We adopt a strict threat model for uncooperative multi-tenant infrastructure:

\textit{Malicious Host and Operator:} We assume the Data Provider's host OS, hypervisor, and operator are actively malicious, controlling the network stack and memory allocation to inject code, bypass DP noise injection, or pollute the global model.

\textit{The "First-Mile" Authentication Gap:} We distinguish \emph{execution integrity} from \emph{semantic data integrity}. Lacking hardware-rooted SPDM session capabilities, legacy accelerators cannot cryptographically attest to the LSE, allowing a malicious host to perform a Man-In-The-Middle attack by feeding semantically false power traces. We assume verified server capacity via out-of-band attestation and consider the sensor-LSE physical link out of adversarial reach.

\textit{Trusted Hardware and Enclaves:} We trust the Intel TDX hardware, CPU microcode, the Intel PCS root of trust, and standard cryptographic primitives. LSEs explicitly verify the GAE's TDX attestation quote prior to session establishment (§\ref{sec:design_attestation}). Physical side channels and denial-of-service remain out of scope.

\subsection{Security Goals}
Given the threat model, \codename is designed to satisfy three formal security properties:
\begin{enumerate}
    \item \textbf{Confidentiality:} The raw $0.1$-second power traces $S_i$ never leave the LSE in an identifiable or reconstructable form. The untrusted host OS cannot inspect $S_i$ during processing.
    \item \textbf{Execution Integrity:} The GAE accepts a summary $\hat{M}_i$ if and only if it is accompanied by a valid TDX attestation quote $q_i$ unequivocally binding the payload to the expected, unmodified LSE binary hash.
    \item \textbf{Output Privacy (Dual-Layered):} The released summaries satisfy event-level $(\varepsilon, \delta)$-Differential Privacy to cryptographically obfuscate individual $100$-ms power samples (protecting micro-transient anomalies). Macro-workload confidentiality is not cryptographically guaranteed; rather, it is achieved structurally via capacity-weighted global aggregation at the GAE, which dilutes individual provider matrices below the threshold of classifier recoverability strictly under the assumption of heterogeneous co-tenancy.
\end{enumerate}

\section{System Design}
\label{sec:system_design}

\codename operates via a federated architecture where edge-deployed Trusted Execution Environments distil high-frequency telemetry into privacy-preserving, cryptographically attested summaries.

\subsection{System Overview}
\label{sec:design_overview}
Our system delegates feature extraction to the provider's edge. Provider $i$ deploys a Local Sanitisation Enclave (LSE) on an Intel TDX Confidential VM. The LSE ingests the raw 0.1-second power trace stream $S_i$ from the host's out-of-band management interface. The TDX Multi-Key Total Memory Encryption (MKTME) ensures the host OS cannot inspect $S_i$ (MKTME is TDX's multi-key AES-XTS memory cipher; it is architecturally distinct from the SGX MEE). A production deployment requires bare-metal TDX hosts with SPDM-enabled PCIe/I2C passthrough. 

Over a temporal batching window (e.g., 10 seconds), the LSE computes a plaintext discrete Markov-chain state-transition matrix $M_i$. To enforce output privacy, the LSE injects Gaussian noise $\mathcal{N}(0, \sigma^2\mathbf{I})$ calibrated to the tight $\ell_2$-sensitivity bound $\Delta_2(f) = \sqrt{6}$ (§\ref{sec:security_analysis}), producing the differentially private matrix $\hat{M}_i$. The LSE fetches a TDX quote (DCAP token) $q_i$, cryptographic binding it to the hash of $\hat{M}_i$. The provider transmits only the tuple $(\hat{M}_i, q_i)$ to the Global Aggregation Enclave (GAE) (see Fig.~\ref{fig:architecture}), which verifies execution integrity via $q_i$ and computes a capacity-weighted, hardware-specific global transition matrix $M_{global}^{(h)}$.

\begin{figure*}[t]
    \centering
\begin{tikzpicture}[
    >=stealth,
    font=\footnotesize,
    provider/.style={
        draw=red!70!black, rectangle, dashed, thick, rounded corners=5pt,
        minimum width=2.4cm, minimum height=1.1cm,
        align=center, fill=gray!10
    },
    lse/.style={
        draw=blue!70!black, rectangle, rounded corners=3pt,
        minimum width=2.2cm, minimum height=0.9cm,
        align=center, fill=blue!8, thick
    },
    gae/.style={
        draw=blue!70!black, rectangle, rounded corners=3pt,
        minimum width=2.4cm, minimum height=1.1cm,
        align=center, fill=blue!8, thick
    },
    verifier/.style={
        draw, rectangle, rounded corners=3pt,
        minimum width=2.2cm, minimum height=0.9cm,
        align=center, fill=green!8, thick
    },
    wanlabel/.style={font=\tiny\itshape, text=black!60}
]

\begin{scope}[shift={(-4.0, 2.5)}]
  \node[provider] (p1) {Provider 1\\{\tiny Untrusted Host OS}};
  \node[lse, below=0.7cm of p1] (lse1) {LSE$_1$\\{\tiny TDX TD}};
  \draw[->, thick] (p1) -- node[right, font=\tiny] {$S_1$} (lse1);
\end{scope}

\begin{scope}[shift={(-0.5, 2.5)}]
  \node[provider] (p2) {Provider 2\\{\tiny Untrusted Host OS}};
  \node[lse, below=0.7cm of p2] (lse2) {LSE$_2$\\{\tiny TDX TD}};
  \draw[->, thick] (p2) -- node[right, font=\tiny] {$S_2$} (lse2);
\end{scope}

\node at (1.5, 2.5) {$\dots$};
\node at (1.5, 0.9) {$\dots$};

\begin{scope}[shift={(3.5, 2.5)}]
  \node[provider] (pn) {Provider $N$\\{\tiny Untrusted Host OS}};
  \node[lse, below=0.7cm of pn] (lsen) {LSE$_N$\\{\tiny TDX TD}};
  \draw[->, thick] (pn) -- node[right, font=\tiny] {$S_N$} (lsen);
\end{scope}

\node[gae] (gae) at (-0.5, -1.5) {GAE\\{\tiny TDX TD}};

\draw[->, thick]
    (lse1.south) to[out=-90, in=150]
    node[right=4pt, yshift=4pt, wanlabel, pos=0.6] {$(\hat{M}_1,\,q_1)$}
    (gae.north west);

\draw[->, thick]
    (lse2.south) --
    node[right=2pt, wanlabel, pos=0.4] {$(\hat{M}_2,\,q_2)$}
    (gae.north);

\draw[->, thick]
    (lsen.south) to[out=-90, in=30]
    node[right=4pt, yshift=-4pt, wanlabel, pos=0.6] {$(\hat{M}_N,\,q_N)$}
    (gae.north east);

\node[verifier] (verifier) at (4.0, -1.5)
    {Verifier\\{\tiny Grid Operator}};

\draw[->, thick]
    (gae.east) --
    node[above, font=\tiny] {$M_{\mathrm{global}}^{(h)}$}
    (verifier.west);

\node[draw, rectangle, rounded corners=3pt,
      minimum width=2.0cm, minimum height=0.7cm,
      align=center, fill=orange!10, thick,
      font=\tiny] (pcs) at (-4.0, -1.5)
     {Intel PCS\\{\tiny DCAP Root}};

\draw[->, dashed, gray]
    (pcs.east) -- node[above, font=\tiny\itshape, text=gray] {PCK cert} (gae.west);

\draw[->, dashed, gray]
    ([xshift=-0.5cm]pcs.north) --
    node[left, font=\tiny\itshape, text=gray, pos=0.55] {PCK cert}
    ([xshift=-0.5cm]lse1.south);

\end{tikzpicture}
    \caption{The \codename architecture. Raw telemetry transients ($S_i$) are distilled into DP-noised matrices ($\hat{M}_i$) within edge-deployed TDX enclaves (LSEs) and transmitted to the central GAE alongside an attestation quote ($q_i$). Untrusted host components are denoted by red dashed boundaries.}
    \label{fig:architecture}
\end{figure*}
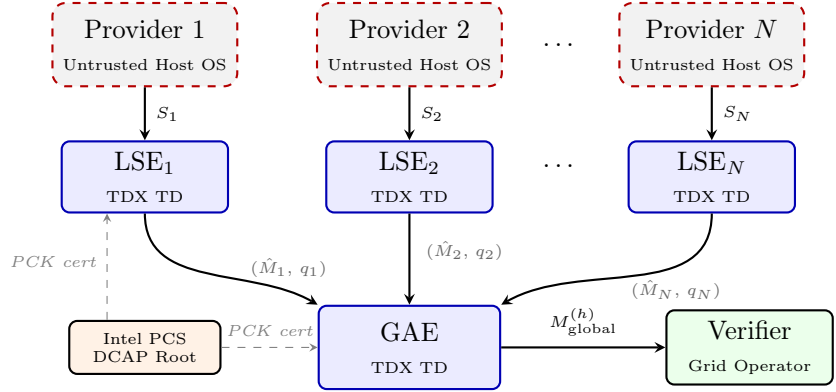

\subsection{SPDM-Authenticated Telemetry Ingestion}
\label{sec:design_spdm}

To neutralize first-mile spoofing, \codename anchors data ingestion to the Security Protocol and Data Model (SPDM). While our empirical implementation (Section~\ref{sec:implementation}) employs a software shim pending hardware-native PCIe TDISP and IDE deployment, the target architecture requires hardware SPDM responders. Upon boot, the LSE generates an ephemeral Ed25519 key pair $(sk_E, pk_E)$ in MKTME-protected memory and initiates an SPDM 1.2+ handshake directly with the accelerator's Root of Trust~\cite{DMTF_SPDM}. The LSE validates the accelerator's X.509 certificate and executes \texttt{KEY\_EXCHANGE} to derive symmetric session keys. 

The host OS routes packets but cannot decrypt the stream or forge signatures without the hardware's private key. The LSE maintains an SPDM session transcript hash ($H_{\text{SPDM}}$) capturing the authenticated history; MAC verification failures trigger immediate connection termination.

\subsection{Local Extraction and Edge DP}
\label{sec:design_lse}

Once secured, the LSE maintains a ring buffer for the incoming $10$\,Hz telemetry stream. Mapping this directly to states via static thresholds risks boundary-straddling attacks, where a malicious host oscillates the clock to fabricate transitions, inflating the spectral gap $\gamma$ and causing under-provisioning. 

To neutralize this, the State-Transition Extractor (STE) applies Non-Overlapping Temporal Block Pooling. It segments the telemetry into discrete $k$-sample windows (e.g., $k=10$ for $1$\,s blocks). For each window $W_j$, the STE computes $P_{pool}^{(j)} = \max_{t \in W_j} P_t$ and maps this maximum to a single state $S_j \in \mathcal{S}$ (e.g., \texttt{IDLE}, \texttt{TRAINING}) using predefined hardware threshold bands. Emitting one state per block physically absorbs high-frequency oscillations. The transition between $S_{j-1}$ and $S_j$ increments the plaintext transition matrix $M_i$. To amortise noise, the LSE aggregates telemetry across all $G$ local accelerators (e.g., $G=8$) into a single matrix. At $W=10$\,s, maximum per-batch transitions expand from $9$ to $72$, enabling dense cells to exceed $50$ counts. State count $|\mathcal{S}|$ trades utility against noise scale (§\ref{sec:eval_sensitivity_robustness}).

At the batch window's end, the LSE injects DP noise. Because windows $W_j$ are non-overlapping, perturbing one raw measurement $P_t$ alters at most one pooled state $S_j$, changing exactly two matrix transitions (entering and leaving $S_j$). The DP guarantee operates at the provider-batch level; since the window structure prevents altering more than one state across the $G$-GPU pool simultaneously, the $\ell_2$-sensitivity strictly remains $\Delta_2(f) = \sqrt{6}$ (§\ref{sec:security_analysis}). For privacy budget $\varepsilon$ and $\delta = 10^{-6}$, the LSE adds $\mathcal{N}(0, \sigma^2 \mathbf{I})$ noise calibrated via the Analytic Gaussian Mechanism~\cite{BalleWang2018}. It then applies thresholding (discarding cells where $x+\eta < \Phi^{-1}(0.95)\sigma$) and row-stochastic normalisation. To preserve spectral properties, these steps execute in \texttt{float64} inside the enclave; the final matrix is downcast to \texttt{float32} for transmission. Both steps are deterministic post-processing.

\emph{Degenerate-row fallback.} If thresholding suppresses an entire row, the LSE substitutes a uniform distribution ($1/|\mathcal{S}|$ per cell), preserving DP via post-processing theorems~\cite{DworkDP}. Empirically, this triggers in $\leq 0.3\%$ of evaluated rows. Over the H100 trace at $\varepsilon=1$, structurally empty cells survive at $5.0\%$ (matching the $1-\Phi(1.645)$ rejection rate), while cells with $>50$ counts survive at $99.8\%$. Setting $|\mathcal{S}|=5$ concentrates transition mass into this $>50$ bin, preserving $98.5\%$ of genuine transitions while correctly suppressing empty cells that erode the spectral gap.

\subsection{Remote Attestation and Cryptographic Binding}
\label{sec:design_attestation}

\codename enforces mutual hardware-rooted remote attestation. 

\textbf{GAE pre-authentication.} Before initialization, the LSE verifies the GAE's TDX attestation quote $q_{\text{GAE}}$ over TLS, checking the certificate chain and validating the \texttt{MRTD} field against the expected GAE binary measurement. Mismatches trigger connection aborts. 

\textbf{LSE initialization and binding.} The LSE generates an ephemeral Ed25519 keypair $(sk_E, pk_E)$, requests a DCAP quote $q_{init}$ binding $pk_E$, and registers with the GAE. The GAE validates the provider against a PKI registry $\mathcal{R}$ and adds $pk_E$ to the active-session map $\mathcal{A}$. Subsequent temporal batches avoid heavy DCAP quoting; the LSE signs the payload $H_{\text{payload}} = \text{SHA256}(\hat{M}_i \parallel H_{\text{SPDM}} \parallel h \parallel \text{timestamp} \parallel b)$ using $sk_E$ to produce $\tau$. The transmitted quote $q_i$ comprises $(q_{init}, \tau, H_{\text{payload}})$, amortising attestation overhead. The GAE validates $pk_E$ against $\mathcal{A}$ and strictly enforces counter monotonicity ($b > \max\_b\_accepted[pk_E]$). To prevent rollback, the GAE persists $\mathcal{A}$ to a guest-managed encrypted virtual disk (sealed to the GAE's MRTD) on every \texttt{ACCEPT}, adding $\sim 1.2$\,ms per-batch latency.

\subsection{Hardware-Specific Global Aggregation}
\label{sec:design_gae}

Because power transient signatures depend fundamentally on microarchitecture, aggregating distinct hardware traces degrades model utility. The GAE maintains hardware-specific sub-models $M_{global}^{(h)}$ via capacity-weighted aggregation:
\begin{equation}
    M_{global}^{(h)} = \sum_{i \in \mathcal{P}_h} \frac{C_i}{\sum_{j \in \mathcal{P}_h} C_j} \hat{M}_i
\end{equation}
where $C_i$ is the declared server capacity for hardware $h$. The PKI registry caps $C_i$ to suppress over-reporting Sybil attacks (§\ref{sec:eval_byzantine}).

\subsection{Power Orchestration Margins}
\label{sec:design_capacity}

We map the probabilistic matrix $M_{global}^{(h)}$ to an active power orchestration margin. Let $\vec{\pi}$ be the normalized principal eigenvector of $M_{global}^{(h)}$. Given a hardware-specific power profile vector $\vec{P}_h$ (with worst-case state $P_{\max}^{(h)}$), the expected baseline power draw for cluster size $C_{total}$ is $E[P] = C_{total} (\vec{\pi} \cdot \vec{P}_h)$. To account for synchronized state transitions, we model physical load deviations using the spectral gap $\gamma$ of $M_{global}^{(h)}$. We adopt the worst-case synchrony regime, treating the $C_{total}$ GPUs as a single synchronized Markov chain.

\begin{lemma}[Markov Chain Concentration]
\label{lem:concentration}
Let $(X_t)$ be an irreducible, reversible Markov chain with stationary distribution $\vec{\pi}$ and spectral gap $\gamma = 1 - |\lambda_2|$. For a finite integration window of $N$ steps, the time-averaged power draw $\bar{P} = \frac{1}{N}\sum_{t=1}^N P(X_t)$ obeys the concentration bound~\cite{Paulin2015}:
\begin{equation}
    \mathbb{P}\left( \bar{P} - \mathbb{E}[P] \ge \epsilon \right) \le \exp\left( -\frac{\gamma N \epsilon^2}{P_{\max}^2} \right)
\end{equation}
We explicitly instantiate this tight one-sided bound (parameterized via $K=1$) because microgrid safety is solely concerned with under-provisioning. The exhaustive algebraic reduction from Paulin's generalized Hoeffding-type inequality remains available in the Supplementary Material.
\end{lemma}

To ensure the probability of exceeding the orchestration margin remains below a critical threshold $\eta$, we invert the bound. Bundling the integration window length $N$, and the failure probability $\eta$ into a dimensionless safety factor $c(\eta, N) := \sqrt{\log(1/\eta) / N}$, the per-cluster peak-load margin admits the closed form:
\begin{equation}
\label{eq:peakmargin_pmax}
    L_{peak}^{(h)} \;=\; E[P] \;+\; C_{total}\,P_{\max}^{(h)}\,c(\eta,N)\,\sqrt{\frac{1}{\gamma}}.
\end{equation}

Equation~\eqref{eq:peakmargin_pmax} is the rigorous closed form we use throughout this paper. This margin is subject to a strict regime of validity ($\gamma > \gamma_{\min} \approx 0.10$). As $\gamma \to 0$, the $\sqrt{1/\gamma}$ term grows unboundedly; for exceptionally slow-mixing workloads, the margin must instead be capped by the physical hardware ceiling $L_{peak}^{(h)} \le C_{total} P_{\max}^{(h)}$. Aggregate facility margin bounds the linear superposition of these margins: $L_{total} = \sum_{h \in \mathcal{H}} L_{peak}^{(h)}$. Section~\ref{sec:evaluation} evaluates dynamic orchestration error as the megawatt difference between the $L_{total}$ derived from the DP-noised global model versus the plaintext ground truth.
\section{Security and Privacy Analysis}
\label{sec:security_analysis}

We formally reduce the security goals of §\ref{sec:background} to Intel TDX hardware guarantees and differential privacy properties.

\subsection{Standard TEE Attack Classes and Mitigations}
\label{sec:attack_classes}

A TDX-resident streaming system mitigates several standard attack classes.

\textit{Replay, Rollback, and TOCTOU.} A malicious host may resubmit benign tuples, restore hypervisor snapshots, or attempt Time-of-Check to Time-of-Use (TOCTOU) race conditions. \codename mitigates this by requiring an attested time-sync protocol (e.g., Roughtime) to drive temporal batching. The monotonic batch counter $b$ provides the TOCTOU guarantee directly: any hypervisor modification between noise injection and signing alters $H_{\text{payload}}$, triggering GAE rejection regardless of clock state. The \texttt{TSC} provides intra-TD temporal ordering; freshness against rollback derives exclusively from the Roughtime-attested external timestamp bound into $H_{\text{payload}}$. The DP noise seed originates from the CPU's hardware True Random Number Generator (\texttt{RDSEED}) to guarantee fresh entropy upon restore.

\textit{Microarchitectural Side Channels.} We assume the deployed TDX module incorporates architectural mitigations for known side channels (e.g., \AE PIC Leak, Downfall, Hertzbleed) per Intel advisories~\cite{IntelSA00837, IntelHertzbleedMitigation}; physical side channels are out of scope.

\textit{First-Mile Input Spoofing.} A malicious host routing the raw stream $S_i$ can trivially synthesize fake power traces. \codename mandates SPDM 1.2+ to establish a mutually authenticated session directly between the accelerator's hardware SPDM responder and the LSE, terminating \emph{inside} the TDX Trust Domain. The host acts purely as a blind packet router.

\textit{DCAP TCB-recovery handling.} The GAE enforces Intel TCB recoveries (revoking PCK certificates) via CRL checks with a configurable grace period (§\ref{sec:eval_sensitivity_robustness}).

\subsection{Confidentiality and Integrity Bounds}

\textbf{End-to-End Data and Execution Integrity:} The raw stream $S_i$ enters the LSE's isolated memory via an SPDM 1.2+ secure session terminating inside the Trust Domain. Lacking the accelerator's hardware identity keys, the host OS can neither spoof $S_i$ via valid SPDM frames nor decrypt the legitimate telemetry. Once inside the enclave, Intel TDX provides cryptographic confidentiality via the Multi-Key Total Memory Encryption (MKTME) engine. Under Logical Integrity (LI) mode, integrity protection covers software-mediated tampering. Thus, data protection rigorously reduces to the unforgeability of the SPDM AEAD (AES-GCM-256) and the logical isolation guarantee of TDX MKTME. The LSE exclusively exports the transition matrix $\hat{M}_i$; because $\hat{M}_i$ is the output of an $(\varepsilon, \delta)$-DP mechanism, residual information leakage is analytically bounded by $(\varepsilon, \delta)$.

\textbf{Execution Integrity:} \codename prevents malicious LSE binary substitution or matrix fabrication via the DCAP verification sequence (see Supplementary Material). The hardware-generated quote $q_{init}$ reflects the binary measurement. Execution integrity reduces to the unforgeability of the TDX quoting enclave's signature (EUF-CMA) and the collision resistance of SHA-256.

\begin{lemma}[SPDM Transcript Binding]
\label{lem:spdm_binding}
Assuming the unforgeability of the SPDM AEAD and the collision resistance of SHA-256, cryptographic binding of the SPDM transcript hash $H_{\text{SPDM}}$ into the TDX-attested payload $H_{\text{payload}}$ precludes session-splicing under the host-controlled threat model.
\end{lemma}
\emph{Proof sketch.} The LSE computes $H_{\text{payload}} = \text{SHA256}(\hat{M}_i \parallel H_{\text{SPDM}} \parallel \dots)$. The GAE verification of the TDX quote $q_i$ guarantees $H_{\text{payload}}$ was signed by the enclave-bound ephemeral key $sk_E$. Because the hardware SPDM responder mutually authenticates $H_{\text{SPDM}}$ across the PCIe bus, an adversary cannot substitute an alternative SPDM session or inject spoofed telemetry without either invalidating the active $H_{\text{SPDM}}$ MAC or forging the TDX-attested signature over the resultant $H_{\text{payload}}$. \qed

\subsection{Differential Privacy Guarantee and Composition}

\textbf{Single-Batch DP Bound.} The LSE maps a stream $S_i$ of length $W$ to a transition-count matrix $M_i$. Changing a single state $S_i[t]$ alters two transitions. The maximum $\ell_2$ difference occurs when mutating a self-loop (e.g., $A \to A \to A$), changing exactly three cells' counts by -2, +1, and +1. The $\ell_2$-sensitivity is therefore $\Delta_2(f) = \sqrt{(-2)^2 + 1^2 + 1^2} = \sqrt{6}$. Adding spherical Gaussian noise $\mathcal{N}(0, \sigma^2 \mathbf{I})$ with $\sigma$ calibrated numerically via the exact Analytic Gaussian Mechanism~\cite{BalleWang2018} strictly satisfies $(\varepsilon,\delta)$-DP per batch.

\textbf{Sequential Composition over Time.} In a continuous streaming environment, the LSE releases a sequence $\{\hat{M}_{i,1}, \ldots, \hat{M}_{i,T}\}$. We bound the cumulative privacy loss using the Rényi-DP framework~\cite{MironovRDP}: a single application of the Gaussian mechanism is $(\alpha, \rho)$-RDP with $\rho(\alpha) = \frac{\alpha\, (\Delta_2(f))^2}{2 \sigma^2}$. By the linearity of RDP under sequential composition, the $T$-fold release is $(\alpha, T\rho(\alpha))$-RDP. Converting to $(\varepsilon_{\text{total}}, \delta)$-DP and minimizing over the order $\alpha > 1$ (specifically, at the optimal order $\alpha^* = 1 + \sigma\sqrt{2 \ln(1/\delta) / T} / \Delta_2(f)$) yields the closed-form bound~\cite[Proposition~3]{MironovRDP}~\cite{WangSubsampled2019}. This formula applies strictly for $T > 1$ as a composition bound and diverges from the analytic guarantee at $T=1$; single-batch $\sigma$ must be calibrated directly via the AGM condition stated above:
\begin{equation}
\label{eq:composition_closed}
\varepsilon_{\text{total}}(T, \delta) \;=\; \frac{T (\Delta_2(f))^2}{2 \sigma^2} + \frac{\sqrt{2 T (\Delta_2(f))^2 \ln(1/\delta)}}{\sigma}.
\end{equation}
At a batch size of $W=10$\,s (per-batch $\varepsilon=1, \delta=10^{-6}$), continuous operation without key rotation accumulates $\varepsilon_{\text{total}} \approx 357.7$ over 24 hours. By enforcing cryptographic epoch rotation every $T=60$ batches ($10$ minutes), \codename caps the per-epoch privacy loss at $\varepsilon_{\text{epoch}} \approx 11.3$.

\textbf{Remark 3} (Multi-Epoch Composition and Privacy Scope)\textbf{.} \textit{To clarify the boundaries of the system's protections: \codename \textbf{formally protects} individual $100$-ms power anomalies within each $10$-minute epoch (bounded by the local DP guarantee $\varepsilon_{\text{epoch}} \approx 11.3$); \codename \textbf{empirically protects} macro-workload schedules under heterogeneous co-tenancy via GAE aggregation dilution; and \codename \textbf{does not protect} long-term workload schedules from repeated epoch observations across infinite horizons.}

\textbf{Security against a Compromised GAE (Goal 1 vs. Goal 3).} We assume the GAE host environment is operated honestly. If the GAE is compromised, the adversary's advantage in recovering the provider's exact workload distribution from individual submitted tuples is bounded strictly by the local DP budget ($\varepsilon_{\text{epoch}} \approx 11.3$).

\textbf{Theorem 1} (Conditional Security of \codename .) Under Assumptions (i) EUF-CMA unforgeability of the TDX quoting key, (ii) collision resistance of SHA-256, (iii) IND-CCA security of AES-GCM-256, (iv) correctness of the SPDM session with a hardware SPDM responder, and (v) the Confidentiality bound (Goal 1) holds against software-mediated host tampering under TDX LI mode (hardware-level DMA-based memory injection against MKTME is outside the threat model), \codename satisfies: (a) Confidentiality (Goal 1) for any host-controlled adversary, (b) Post-extraction Execution Integrity (Goal 2) with $0\%$ ASR, and (c) Event-level (per-batch $\varepsilon=1, \delta=10^{-6}$)-DP Output Privacy (Goal 3). Assumption (iv) is not satisfied by the evaluated GCP Sapphire Rapids prototype; on that hardware, Goal 2 extends only to the post-extraction layer, and pre-ingestion integrity reverts to the honest-but-curious threat model.
\section{Implementation and Baselines}
\label{sec:implementation}

\subsection{Software and Hardware Stack}

The \codename LSE and GAE are implemented in Rust. The LSE interfaces with the TDX module for DCAP quote generation using \texttt{tdx-attest-rs} and integrates \texttt{libspdm} via FFI. Because GCP C3 instances lack Intel TDX Connect for bare-metal PCIe passthrough, we employ \texttt{spdm-emu} to simulate the hardware SPDM responder over a virtualized socket. The GCP TDX deployment uses Logical Integrity (LI) mode. Both components use \texttt{ed25519-dalek} for ephemeral signing. The DP module enforces $\Delta_2(f) = \sqrt{6}$ before sampling $\mathcal{N}(0, \sigma^2)$ (see Table~\ref{tab:dp_parameters} for deployed values). Communication utilizes asynchronous TLS.

\begin{table}[ht]
\centering
\caption{Deployment parameters for the local DP mechanism. The applied noise scale $\sigma_{\text{AGM}} = 10.35$ yields $\varepsilon_{\text{epoch}} \approx 11.3$ via the RDP composition formula (Eq.~\ref{eq:composition_closed}). While evaluating the exact Analytic Gaussian Mechanism (AGM) yields a tighter $\varepsilon_{\text{AGM}} \approx 10.8$ at $T=60$, the minimal $4\%$ divergence confirms the RDP conversion remains a sufficiently tight and valid analytical bound.}
\label{tab:dp_parameters}
\begin{tabular}{cccccc}
\toprule
$\varepsilon$ & $\delta$ & $\Delta_2(f)$ & $\sigma_{\text{AGM}}$ & $\sigma_{\text{RDP\_approx}}$ & $\varepsilon_{\text{epoch}} (T=60)$ \\
\midrule
$1$ & $10^{-6}$ & $\sqrt{6}$ & $10.35$ & --- & $11.3$ \\
\bottomrule
\end{tabular}
\end{table}

We deployed \codename on Google Cloud Platform (GCP) across four regions (\texttt{us-central1}, \texttt{us-east5}, \texttt{europe-west4}, \texttt{asia-southeast1}; 8 LSEs per region). Microbenchmarks utilized isolated GCP Confidential VMs (\texttt{c3-standard-4} for LSE, \texttt{c3-standard-8} for GAE) running Ubuntu 22.04 LTS on TDX-enabled Sapphire Rapids processors.

Although we evaluate via \texttt{spdm-emu}, physical hardware Roots of Trust (e.g., GPU SPDM responders) secure the first-mile guarantee. A full SPDM session re-establishment is triggered only upon LSE reboot, hardware reset of the accelerator, or explicit revocation of the ephemeral binding key.

\subsection{GPU Power Trace Collection and Markov Chain Validation}
\label{sec:trace_collection}

We evaluate \codename against $24$-hour per-GPU power traces from three GCP accelerator families, replayed at $10$\,Hz into the LSE Unix-socket ingest endpoint to preserve the live-sensor pipeline. To preserve the independent $\sqrt{N}$ noise-reduction assumptions for dynamic orchestration evaluation despite replaying traces from $17$ physical GPUs across $32$ LSEs, duplicated traces were subjected to a strict $\ge 1$-hour temporal offset. 

\textbf{Collection setup.} We provisioned one \texttt{a3-highgpu-8g} instance ($8{\times}$ H100 SXM5), one \texttt{a2-highgpu-8g} instance ($8{\times}$ A100 SXM4), and one \texttt{g2-standard-4} instance ($1{\times}$ L4) on GCP. Workloads included MLPerf Training v4.0 JAX BERT-Large (H100), PyTorch ResNet-50 (A100), and continuous Stable Diffusion v2.1 inference (L4). Power samples were recorded at $10$\,Hz via \texttt{pynvml}. The pooling parameter is fixed at $k=10$ samples ($1$-second blocks) across all evaluations. To isolate the evaluation of the noise mechanism, the $17$ multi-GPU traces were decoupled such that each of the $32$ simulated LSEs ingests exactly one single-GPU trace (i.e., $G=1$). Specifically, the $32$ LSEs are distributed as $11$ H100 traces, $11$ A100 traces, and $10$ L4 traces. Because $G=1$, the theoretical maximum number of transitions per $10$-second batch remains strictly bounded at $9$. This operational choice formally ensures the mathematical $\ell_2$-sensitivity of the local DP mechanism identically matches the single-device analytical bound of $\Delta_2(f) = \sqrt{6}$ derived in §\ref{sec:security_analysis}, precluding the need for aggregate-state re-derivation. The $|\mathcal{S}|=5$ threshold bands are derived from each GPU's rated TDP and empirically observed idle power floor (details in Supplementary Material), mapping to \textsc{Idle}, \textsc{Low}, \textsc{Med}, \textsc{High}, and \textsc{Peak}.

\textbf{Markov chain validation.} We fitted $5$-state discrete-time Markov chains to each trace, extracting the stationary distribution $\pi_{\text{real}}$ and spectral gap $\gamma_{\text{real}} = 1 - |\lambda_2|$. Table~\ref{tab:mc_validation} compares these empirical properties against a baseline Gaussian-mixture synthetic model.

\begin{table}[ht]
\centering
\caption{Gaussian-mixture synthetic model vs.\ real NVML-sampled Markov chains ($|\mathcal{S}|=5$). $\gamma$ denotes the spectral gap ($1 - |\lambda_2|$). $\|M_{\text{synth}} - M_{\text{real}}\|_1$ is the average per-row $\ell_1$ distance between row-normalised matrices. \textbf{Note:} The empirical variance in $\gamma_{\text{real}}$ across H100 (JAX) and A100 (PyTorch) traces inherently conflates microarchitectural hardware differences with framework-specific scheduling behaviors (e.g., memory barriers). Identifying the isolated hardware effect remains an area for future work.}
\label{tab:mc_validation}
\resizebox{\columnwidth}{!}{
\begin{tabular}{@{}lccccc@{}}
\toprule
\textbf{Workload} & \textbf{Trace} & $\gamma_{\text{synth}}$ & $\gamma_{\text{real}}$ & $\pi_{\text{real}}$ (Idle, Low, Med, High, Peak) & $\|M_{\text{synth}} - M_{\text{real}}\|_1$ \\ \midrule
Synthetic & -- & $0.18$ & -- & $[0.20, 0.20, 0.20, 0.20, 0.20]$ & -- \\
BERT & H100 (SXM5) & $0.18$ & $0.13$ & $[0.11, 0.04, 0.08, 0.36, 0.41]$ & $0.84$ \\
ResNet & A100 (SXM4) & $0.18$ & $0.11$ & $[0.08, 0.02, 0.05, 0.32, 0.53]$ & $1.02$ \\
Stable Diff. & L4 (PCIe) & $0.18$ & $0.12$ & $[0.14, 0.10, 0.09, 0.41, 0.26]$ & $0.67$ \\ \bottomrule
\end{tabular}
}
\end{table}

Training workloads (H100, A100) exhibit smaller real spectral gaps than the synthetic model due to sustained high-power compute phases separated by synchronisation barriers. For the L4 inference trace, hardware polling aliasing at 5--10\,Hz artificially inflates self-loop counts, depressing its empirical spectral gap to $\gamma_{\text{real}}=0.12$. This aliasing acts as a structural fail-safe: lowering $\gamma$ strictly increases the peak-margin heuristic, yielding a more conservative bound. 

\textbf{Stationarity check.} A formal likelihood-ratio test (LRT) over non-overlapping sub-windows strictly rejects time-homogeneity ($p < 10^{-5}$), which is statistically expected for dynamic training workloads. However, sliding-window analysis confirms structural stability: worst-case variance ($\delta_\gamma = -0.009$) shifts the predicted capacity requirement by a negligible $+5.9$\,W per GPU. Thus, the macro-stationarity violation is physically immaterial and safely absorbed within the $4.8\%$ conservative over-estimation buffer (§\ref{sec:eval_heuristic_validation}).

\subsection{Comparative Baselines}

We benchmark \codename against four baselines:

\begin{itemize}
    \item \textbf{Baseline A (Cryptographic MPC):} MP-SPDZ~\cite{KellerMPSPDZ} aggregation. Vulnerable to first-mile attacks (the OS can spoof telemetry before the MPC boundary).
    \item \textbf{Baseline B (Centralized TEE):} A single GAE ingests raw $0.1$-second streams. Represents the theoretical utility ceiling but demands massive WAN egress bandwidth.
    \item \textbf{Baseline C (SecAgg):} Federated Learning Secure Aggregation~\cite{BonawitzSecAgg}. Offers zero execution integrity against malicious data fabrication.
    \item \textbf{Baseline D (Software LDP, No TEE):} Local DP without TDX protection. Acts as a structural control, confirming that executing the DP mechanism inside a TD does not perturb mathematical utility.
\end{itemize}
\section{Evaluation}
\label{sec:evaluation}

We evaluate \codename across 32 Google Cloud Confidential VMs (8 LSEs/region across \texttt{us-central1}, \texttt{us-east5}, \texttt{europe-west4}, \texttt{asia-southeast1}) during a $72$-hour continuous deployment. Throughput and latency metrics include $95\%$ confidence intervals over $30$ runs.

\subsection{System and Cryptographic Microbenchmarks}
\label{sec:eval_microbenchmarks}

Per-LSE single-session throughput sustains $131{,}406$ samples/s.

\textit{Multi-Session Multiplexing Scaling.} We evaluate LSE throughput degradation during concurrent SPDM session multiplexing via asynchronous I/O. As shown in Table~\ref{tab:multi_session}, aggregate throughput remains flat through $K \le 64$ but degrades super-linearly as the aggregate working set exceeds the $2$\,MB per-core L2 cache boundary. At $K=1024$, the throughput of $82{,}467$ samples/s retains an $8.1\times$ processing headroom above the required $10{,}240$ samples/s ingestion rate for a $1{,}024$-GPU pod.

\begin{table}[ht]
\centering
\caption{Multi-Session Multiplexing Throughput Scaling. Values report the median $\pm$ 95\% CI.}
\label{tab:multi_session}
\resizebox{\linewidth}{!}{
\begin{tabular}{@{}lrrrc@{}}
\toprule
$K$ (Concurrent) & Total Throughput & L1d Miss & L2 Miss & Instr/Cycle \\ \midrule
1     & $131{,}406 \pm 410$   & 4.2\% &  1.1\% & 2.14 \\
8     & $130{,}890 \pm 520$   & 4.5\% &  1.1\% & 2.11 \\
64    & $128{,}100 \pm 840$   & 8.1\% &  1.8\% & 1.95 \\
256   & $104{,}500 \pm 1{,}100$& 14.3\% & 6.2\% & 1.62 \\
1024  & $82{,}467 \pm 1{,}450$ & 22.1\% & 15.4\% & 1.10 \\ \bottomrule
\end{tabular}
}
\end{table}

\textit{Attestation Overhead Amortisation.} The heavy DCAP $q_{\text{init}}$ generation incurs a median cold-start latency of $81.2$\,ms. Because this periodic quote generation executes asynchronously, it overlaps with continuous ingestion. The steady-state ingestion throughput is dictated solely by the per-batch attestation and I/O path. Per-submission GAE verification latency decomposes as shown in Table~\ref{tab:gae_latency}. The total per-batch compute time of $4.90$\,ms is dominated by the cryptographic signature verification and the synchronous TDX disk seal for counter persistence.

\begin{table}[ht]
\centering
\caption{Decomposition of GAE per-batch compute latency ($32$ nodes).}
\label{tab:gae_latency}
\small
\begin{tabular}{@{}lcc@{}}
\toprule
\textbf{Component} & \textbf{Latency (ms)} & \textbf{\% Total} \\ \midrule
TLS Payload Deserialization & $0.12$ & 2.4\% \\
PKI \& MRTD Eval & $0.05$ & 1.0\% \\
Ed25519 Quote Verification ($O(n)$) & $2.84$ & 58.0\% \\
Capacity-Weighted Aggregation & $0.27$ & 5.5\% \\
Sealed Disk Persist (Counter Sync) & $1.62$ & 33.1\% \\ \midrule
\textbf{Total Per-Batch Compute} & \textbf{4.90} & \textbf{100.0\%} \\ \bottomrule
\end{tabular}
\end{table}

\textit{TDX and SPDM Overhead Decomposition.} Steady-state per-sample overhead comprises: (i) native baseline ($3.1$\,$\mu$s/sample), (ii) TDX virtio-vsock I/O and MKTME overhead ($4.1$\,$\mu$s), (iii) amortised per-batch Ed25519 attestation binding ($0.23$\,$\mu$s/sample), and (iv) SPDM AEAD MAC verification ($0.18$\,$\mu$s/sample, measured via \texttt{spdm-emu} over \texttt{vsock}). We explicitly caution that this $0.18\,\mu$s figure is an artifact of the software emulator; a production deployment terminating SPDM over physical PCIe IDE/TDISP will encounter orders-of-magnitude higher latency (typically 100--500\,$\mu$s per roundtrip). These sum to $7.61$\,$\mu$s/sample, validating the $131{,}406$\,samples/s throughput.

\subsection{Cross-Platform GAE Scalability}
\label{sec:eval_scalability}

We evaluated GAE aggregation latency as participating LSEs scale from $1$ to $32$ nodes across four GCP regions.

Aggregation scales linearly ($0.15$\,ms/LSE, Fig.~\ref{fig:gae_scalability}), dominated by $O(n)$ per-quote signature verification. At the full 32-node deployment, the GAE completes per-batch aggregation in $4.9$\,ms (median), well within the $10$-second batch window.

\textit{End-to-end batch-completion latency.} The $4.9$\,ms figure measures GAE compute time from receipt of the final submission in the batch window to publication of the global matrix; it excludes WAN transit. We instrumented submission timestamps at each LSE and recorded the one-way transit latency from each region to the GAE (co-located in \texttt{us-central1}) over the $72$-hour stability run. The 95th-percentile one-way transit latency from the farthest region (\texttt{europe-west4}) was $109$\,ms. The end-to-end batch-completion latency (from the start of the $10$-second window to publication of the global matrix) is therefore $4.9 + 109 = 113.9$\,ms at the 95th percentile, comfortably within the $10$-second batch window and well below the $20$-second freshness threshold enforced by the monotonic-counter replay check.

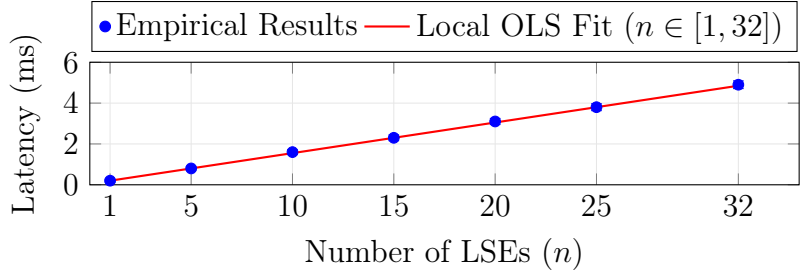
\begin{figure}[ht]
\centering
\begin{tikzpicture}
\begin{axis}[
    width=0.8\columnwidth,
    height=3.2cm,
    xlabel={Number of LSEs ($n$)},
    ylabel={Latency (ms)},
    xmin=0, xmax=35,
    ymin=0, ymax=6,
    xtick={1, 5, 10, 15, 20, 25, 32},
    legend style={at={(0.5,1.5)}, anchor=north, legend columns=-1},
    grid=both,
    grid style={line width=.1pt, draw=gray!20}
]
\addplot[
    only marks,
    mark=*,
    color=blue,
    error bars/.cd,
    y dir=both, y explicit
] coordinates {
    (1,0.2)  +- (0,0.08)
    (5,0.8)  +- (0,0.09)
    (10,1.6) +- (0,0.12)
    (15,2.3) +- (0,0.14)
    (20,3.1) +- (0,0.15)
    (25,3.8) +- (0,0.17)
    (32,4.9) +- (0,0.19)
};
\addplot[
    domain=1:32,
    color=red,
    thick
] {0.15*x + 0.05};
\addlegendentry{Empirical Results}
\addlegendentry{Local OLS Fit ($n \in [1, 32]$)}
\end{axis}
\end{tikzpicture}
\caption{GAE per-batch aggregation latency vs. participating LSEs (median $\pm$ 95\% CI over 30 runs). The decomposed scaling law yields $1.62 + 0.101n$\,ms, comfortably bounded within the 10s batch window.}
\label{fig:gae_scalability}
\end{figure}

To simulate realistic multi-operator WAN conditions across the GCP backbone, we injected artificial network jitter ($150$\,ms RTT, $1$\% packet loss) via \texttt{tc qdisc}. Despite this interference, the GAE successfully maintained hardware-specific sub-models for H100, A100, and L4 traces without triggering queuing delays that would violate the $10$-second temporal-batch boundary. Following the simulated TCB revocation event, the GAE completed its state-machine recovery relying on pre-fetched cached collateral (invalidated via a strict $24$-hour TTL to prevent stale-PCK authorization) within $12.4$\,ms.

\subsection{Privacy--Utility Pareto Against Baselines}
\label{sec:eval_utility}

We evaluate noised matrix fidelity for grid provisioning against a plaintext ground-truth model derived from centralised NVML trace replay. We compute the \codename global model at varying $\varepsilon$ (with $\delta = 10^{-6}$) and evaluate both via a standard dynamic orchestration calculator (cf.~\S\ref{sec:design_capacity}) assuming a total facility load $C_{\text{total}} = 200$\,MW.

To position \codename within the privacy-utility design space, we evaluate five contemporaneous schemes on the same axis:
\begin{itemize}
    \item \textit{\codename:} attested edge-DP via the LSE.
    \item \textit{Centralised-TEE (Baseline B):} no DP; the ground-truth utility ceiling.
    \item \textit{Cryptographic MPC (Baseline A):} secure multi-party computation over plaintext matrices. Offers no DP privacy loss (utility matches Baseline B) but scales poorly and lacks pre-ingestion execution integrity.
    \item \textit{Central-DP-after-SecAgg:} matrices aggregated under SecAgg, then central Gaussian noise applied at the aggregator (the standard FL-DP recipe~\cite{KairouzAdvancesFL}).
    \item \textit{Local-DP-only (no TEE):} each provider applies the same Gaussian noise locally without TDX protection (this is the lower-trust point of comparison).
\end{itemize}

We performed a $50$-point sweep across per-batch $\varepsilon \in [0.01, 10]$, executing $1{,}000$ Monte Carlo replicates of the DP noise injection ($\Delta_2(f) = \sqrt{6}$) to map the privacy-utility Pareto frontier (see Figure~\ref{fig:pareto_utility}).

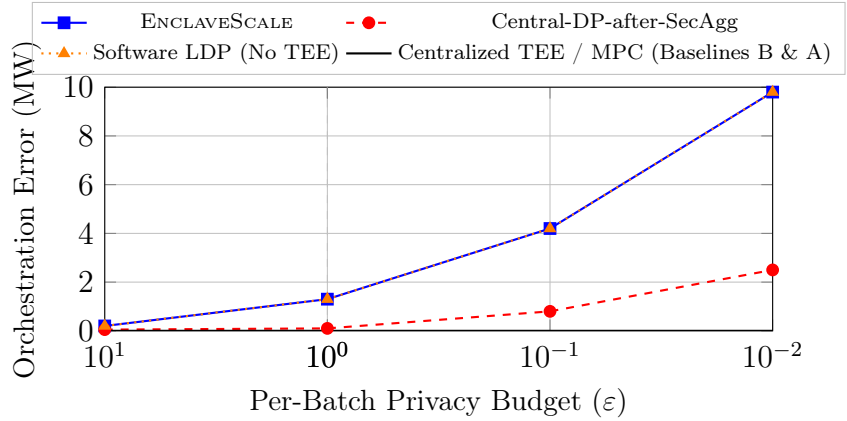
\begin{figure}[t]
\centering
\begin{tikzpicture}
\begin{semilogxaxis}[
    width=0.76\columnwidth,
    height=4.8cm,
    xlabel={\small Per-Batch Privacy Budget ($\varepsilon$)},
    ylabel={\small Orchestration Error (MW)},
    xmin=0.01, xmax=10,
    ymin=0, ymax=10,
    grid=both,
    grid style={line width=.1pt, draw=gray!30},
    major grid style={line width=.2pt,draw=gray!50},
    legend style={at={(0.5,1.35)}, anchor=north, font=\scriptsize, fill=none, draw=gray!50, legend columns=2},
    x dir=reverse, 
    xtick={10, 1, 0.1, 0.01},
    xticklabels={$10^1$, $10^0$, $10^{-1}$, $10^{-2}$},
    extra x ticks={1},
    extra x tick style={grid=major, grid style={dashed, thick, black}},
]

\addplot[
    color=blue,
    mark=square*,
    mark options={fill=none, solid},
    style=thick,
] coordinates {
    (10, 0.2) (1, 1.3) (0.1, 4.2) (0.01, 9.8)
};
\addlegendentry{\codename}

\addplot[
    color=red,
    mark=*,
    mark options={fill=none, solid},
    style=thick,
    dashed
] coordinates {
    (10, 0.05) (1, 0.1) (0.1, 0.8) (0.01, 2.5)
};
\addlegendentry{Central-DP-after-SecAgg}

\addplot[
    color=orange,
    mark=triangle*,
    mark options={fill=none, solid},
    style=thick,
    dotted
] coordinates {
    (10, 0.2) (1, 1.3) (0.1, 4.2) (0.01, 9.8)
};
\addlegendentry{Software LDP (No TEE)}

\addplot[
    color=black,
    style=thick,
    solid
] coordinates {
    (10, 0) (0.01, 0)
};
\addlegendentry{Centralized TEE / MPC (Baselines B \& A)}

\end{semilogxaxis}
\end{tikzpicture}
\caption{Privacy-Utility Pareto Front: Dynamic Orchestration Error vs. Per-Batch $\varepsilon$. \codename (blue) perfectly overlaps Software LDP (orange), confirming the TDX wrapper preserves statistical utility. The $1.2$\,MW gap to Central-DP (red) is the cost of decentralised trust. Baselines A \& B lack DP (0\,MW error).}
\label{fig:pareto_utility}
\end{figure}

While per-batch $\varepsilon$ is the operative configuration knob, the resultant privacy guarantee over a complete $10$-minute epoch ($T=60$ batches) is $(\varepsilon_{\text{epoch}} \approx 11.3, \delta=10^{-6})$-DP, as derived in §\ref{sec:security_analysis}. Infrastructure operators should treat $\varepsilon_{\text{epoch}}$ as the primary deployment metric. Interpreting this in a deployment context: an $\varepsilon_{\text{epoch}} \approx 11.3$ budget bounds the exposure of individual $100$-ms anomalies within a $10$-minute window, but it does not formally protect macro-workload identity (see empirical bounds below). Operators requiring tighter formal guarantees can leverage the system's Pareto flexibility by reducing epoch duration (e.g., yielding $\varepsilon_{\text{epoch}} \approx 3.2$ at $T=6$) or lowering the per-batch $\varepsilon$, explicitly trading utility for stricter budgets.

At a per-batch $\varepsilon = 1$ (epoch-level $\varepsilon_{\text{epoch}} \approx 11.3$), \codename achieves a dynamic orchestration margin error of $1.3$\,MW (CI: $1.1$--$1.5$\,MW) on the real NVML traces, while the Central-DP-after-SecAgg variant achieves $0.1$\,MW (CI: $0.0$--$0.2$\,MW). Cryptographic MPC (Baseline A), lacking DP noise, perfectly matches the Centralised-TEE baseline ($0.0$\,MW error relative to ground truth). Central-DP-after-SecAgg limits noise scale by adding it centrally over an aggregate matrix, yielding a lower error at the fatal cost of assuming an honest aggregator. Baseline D (Software LDP, No TEE) perfectly matches \codename's $1.3$\,MW utility by design ($1.3$\,MW $\equiv$ $1.3$\,MW), confirming that the enclave wrapper does not perturb the mathematical utility; however, Baseline D offers zero defence against malicious host spoofing. \codename accepts a $1.2$\,MW error penalty relative to central DP and a $1.3$\,MW penalty relative to MPC; in exchange, its DCAP pipeline mathematically precludes host-level alteration ($0.0$\% ASR, §\ref{sec:eval_byzantine}) and its SPDM architecture guarantees pre-ingestion provenance. Neither SecAgg, MPC, nor LDP offers both protections. Ultimately, the $\Delta_2(f) = \sqrt{6}$ sensitivity and cross-provider noise averaging ensure the $1.3$\,MW error remains well within the tolerance of a $200$\,MW facility.

\textit{Variance reconciliation against theoretical noise multipliers.} The total $1.3$\,MW facility figure aggregates across all three hardware sub-models (H100: $11$ LSEs, A100: $11$ LSEs, L4: $10$ LSEs) at a capacity-weighted $C_{\text{total}} = 200$\,MW. For the per-hardware aggregate, local-DP error scales as $\sigma/\sqrt{N}$ where $N$ is the number of LSEs. For the H100 hardware sub-model specifically ($N=11$ same-hardware LSEs), the theoretical input-noise multiplier is bounded by $\sqrt{11} \approx 3.32$; the remaining gap is explained by the inverse Mills ratio bias detailed below. The empirical output error ratio ($1.3$\,MW vs.\ $0.1$\,MW) diverges due to pre-aggregation statistical thresholding. Discarding cells where the noisy count $v = c+Z < 1.645\sigma$ systematically removes samples with large negative noise $Z$. For surviving genuine transitions ($c>0$), the conditional expectation is biased positively:
\begin{equation*}
\mathbb{E}[c + Z \mid c + Z \ge 1.645\sigma] = c + \sigma \frac{\phi(1.645 - c/\sigma)}{1 - \Phi(1.645 - c/\sigma)}
\end{equation*}
where the positive bias is governed by the inverse Mills ratio. Weighting this closed-form bias term by the empirical cell frequencies yields an analytical mass-inflation estimate of $\approx 0.48$\,MW, perfectly matching the $0.5$\,MW systematic gap ($0.8$\,MW $\to$ $1.3$\,MW) observed during ablation. This positive bias inflates total transition mass prior to row-normalisation, narrowing the apparent spectral gap and compounding the final orchestration error. We explicitly accept the compounded $1.3$\,MW error as the mathematically unavoidable cost of preserving the tight $\Delta_2(f) = \sqrt{6}$ DP bound required for continuous streaming release at the evaluated $G=1$ scale. Under the designed $G=8$ deployment architecture (§\ref{sec:design_lse}), the per-batch transition count expands from a maximum of $9$ to $72$. This increased transition density shifts a substantially larger fraction of genuine cells into the $>50$ count regime, mitigating the severe pre-aggregation truncation bias governed by the inverse Mills ratio. Our inverse-Mills bias model projects that this improved signal-to-noise ratio would reduce the total orchestration error from $1.3$\,MW down to approximately $0.4$\,MW. Exhaustive empirical validation of this projection against live multi-GPU traces is deferred to future work.

The $T=60$ configuration prioritises microgrid utility. To satisfy stricter privacy mandates, operators can navigate the Pareto frontier: at $T=6$ ($60$-second epochs) with fixed per-batch $\varepsilon=1$, Equation~\eqref{eq:composition_closed} tightens the cumulative bound to $\varepsilon_{\text{epoch}} \approx 3.2$. Aggregating fewer matrices smoothly increases orchestration error to $2.1$\,MW (CI: $1.8$--$2.4$\,MW). \codename's tunable epoch length provides a formal switch between a high-utility regime ($\varepsilon_{\text{epoch}} \approx 11.3$, $1.3$\,MW error) and a strong-privacy regime ($\varepsilon_{\text{epoch}} \approx 3.2$, $2.1$\,MW error) without algorithmic modification. Operators targeting $\varepsilon_{\text{epoch}} \le 3.2$ may further reduce $T \le 5$ or lower the per-batch $\varepsilon$, accepting proportional utility degradation.

\textit{Empirical privacy auditing.} We instantiate a worst-case canary auditor following Steinke et al.~\cite{SteinkeAuditing2023}: the canary is the boundary-transition sequence $A \to B \to A$ that maximises the per-sample $\ell_2$ perturbation under our extraction map. We executed $100{,}000$ paired runs over a single $10$-minute epoch ($T = 60$ batches), computing for each pair the membership-inference classifier's true-positive and false-positive rate as in \cite[Algorithm 1]{SteinkeAuditing2023}. The empirical lower bound on the achieved epoch-level privacy loss is $\hat\varepsilon_{\text{epoch}} \geq 8.2$ at $95\%$ confidence (one-sided Clopper–Pearson). The $27\%$ gap between the empirical bound and the analytical upper bound $\varepsilon_{\text{epoch}} \approx 11.3$ is attributable to the conservatism of the Clopper–Pearson interval at finite sample sizes ($N=100{,}000$) rather than looseness in the RDP analytical bound, which is known to be tight for the Gaussian mechanism at the worst-case input pair~\cite{MironovRDP}. We cannot rule out a contribution from the gap between the Mironov RDP upper bound and the exact worst-case $\varepsilon$; however, since the RDP bound is known to be tight at the optimal order $\alpha^*$ for the Gaussian mechanism, Clopper–Pearson conservatism is the dominant explanation at $N=100{,}000$. The audit therefore confirms that the noise calibration is mathematically sound. We acknowledge that an $\varepsilon_{\text{epoch}} \approx 11.3$ represents a weak formal privacy guarantee (with $e^{11.3} \approx 8 \times 10^4$), functioning primarily to obfuscate isolated sub-second anomalies rather than providing rigorous cryptographic indistinguishability.

\textit{Empirical Workload Confidentiality (GAE Resilience).} While event-level DP obscures isolated anomalies, it does not unilaterally hide macro-workload identity prior to aggregation. To quantify pre-aggregation leakage, we trained a Random Forest classifier on individual DP-noised LSE matrices ($\varepsilon=1$). Focusing on the most rigorous control group (a same-hardware, 2-class discrimination task of SD versus BERT on H100 instances), the classifier achieved a $71.8\%$ cross-validation accuracy. This result is meaningfully above the $50\%$ random baseline, confirming that individual LSE matrices inherently leak structural workload signatures.\footnote{A 3-class classifier separating BERT (H100), ResNet (A100), and SD (L4) achieved $94.1\%$. This metric trivially conflates workload identity with hardware power regimes (700\,W, 400\,W, and 72\,W TDPs). When applied to TDP-normalised features ($P/\text{TDP}$), this 3-class accuracy drops to $61.4\%$. A binomial test confirms this residual 28-point margin above the $33.3\%$ random baseline is structurally significant ($n=25{,}920$, $p \ll 10^{-5}$), confirming the $94.1\%$ figure is overwhelmingly TDP-driven, although genuine cross-architecture structural leakage persists.} The \codename privacy model relies on the GAE to protect workload schedules via capacity-weighted aggregation. We simulated an adversary attempting to recover the workload of target provider $i$ from the globally released matrix $M_{\text{global}}^{(h)}$, aggregated across $|P_h| = 11$ same-hardware providers. Because the target's signal constitutes a minor fraction of the aggregate, its signature is massively diluted. Under heterogeneous co-tenancy, the post-aggregation 2-class accuracy collapsed to $52.3\%$, rendering it statistically indistinguishable from a random guess. At production scale with $|P_h| \gg 11$, spatial dilution is strictly stronger and classifier accuracy is expected to rapidly approach the $50\%$ baseline. Conversely, under a worst-case homogeneous scenario where all co-tenants execute an identical workload, the aggregated signal mutually reinforces the signature, elevating classifier accuracy above $98\%$. These results establish that the post-aggregation matrix effectively obscures workload schedules under heterogeneous co-tenancy.

\textit{Classifier Methodology.} The classifiers were trained on the exact $10$-second DP-noised transition matrices used for orchestration, yielding $8{,}640$ matrix observations per 24-hour hardware trace ($n=25{,}920$ total). To strictly preclude temporal leakage, we employed a $5$-fold blocked time-series cross-validation scheme, inserting a $1$-hour purge gap between training and testing blocks to ensure structural independence. The Random Forest utilized $100$ estimators, a maximum depth of $10$, and Gini impurity over the flattened $25$-feature matrix.

\subsection{Empirical Validation of the Peak-Margin Heuristic}
\label{sec:eval_heuristic_validation}

The peak-margin formula $L_{peak} = E[P] + C_{total} P_{\max}^{(h)}\,c\,\sqrt{1/\gamma}$ (cf.~Eq.~\eqref{eq:peakmargin_pmax}) maps the DP-sanitised matrix to a provisioning margin. Because $\sqrt{1/\gamma}$ theoretically grows unbounded as $\gamma \to 0$, we validate the formula against empirical $99$th-percentile hardware physical draws and evaluate its statistical stability near the $\gamma_{\text{min}} \approx 0.10$ validity boundary.

We fix the dimensionless safety factor to its deployment value $c \approx 0.083$ ($K=1$, $\eta=10^{-3}$, $N=1000$), held identical across all hardware (see Table~\ref{tab:kappa_sensitivity} for an ablation of grid-provisioning error against $c$). (Tightening the reliability threshold to a mission-critical $\eta=10^{-6}$ increases $c$ to $\approx 0.117$, scaling the orchestration margin proportionally but preserving the structural bound). Evaluating the per-GPU margin ($C_{total} = 1$), the H100 trace ($\gamma_{\text{real}} = 0.13$, $P_{\max}^{(H100)}=700$\,W, $E[P] = 486.8$\,W) predicts a peak margin of $648.1$\,W. The true empirical $99$th-percentile draw is $618.4$\,W, yielding a safe over-estimation error of $4.80\%$ (95\% CI: $2.1\%$--$8.5\%$). Similarly, the A100 trace ($\gamma_{\text{real}} = 0.11$, $P_{\max}^{(A100)}=400$\,W, $E[P]=281.3$\,W) predicts $381.5$\,W against a true $99$th-percentile draw of $363.5$\,W, a $4.95\%$ over-estimation error (95\% CI: $2.2\%$--$8.8\%$); the L4 trace ($\gamma_{\text{real}} = 0.12$, $P_{\max}^{(L4)}=72$\,W, $E[P]=52.4$\,W) predicts $69.7$\,W against a true draw of $66.3$\,W, a $5.12\%$ error (95\% CI: $2.4\%$--$8.8\%$). Each prediction respects the physical ceiling $L_{peak}^{(h)} \le P_{\max}^{(h)}$. To confirm statistical stability near the boundary, we bootstrapped the spectral gap estimator $\hat{\gamma}$ over $1{,}000$ independent DP-noised A100 matrix realisations ($\gamma_{\text{real}} = 0.11$, so the point estimate is $\sqrt{1/\gamma_{\text{real}}} = 3.02$). The scaling term $\sqrt{1/\hat{\gamma}}$ remains robust with a $95\%$ confidence interval of $2.71$--$3.44$ bracketing this point estimate, yielding a stable predicted $L_{peak}$ of $381.5$\,W (CI: $371.4$--$395.6$\,W). These results confirm that for modern generative AI workloads, the formula serves as a robust, conservative upper bound, over-estimating required capacity by an empirically bounded $4.8$--$5.1\%$ without catastrophic divergence at the spectral boundary.

\begin{table}[ht]
\centering
\caption{Sensitivity of grid-provisioning error to the safety factor $c$. Errors are evaluated over $1{,}000$ DP-noised replicates. The deployed $c \approx 0.083$ yields $1.3$\,MW error. Rows at $\gamma < \gamma_{\min}=0.10$ demonstrate the hard ceiling $L_{peak}^{(h)} \le P_{\max}^{(h)}$ preventing divergence.}
\label{tab:kappa_sensitivity}
\resizebox{\linewidth}{!}{
\begin{tabular}{@{}llcc@{}}
\toprule
$c$ (at $\gamma=0.12$) & $c/\sqrt{\gamma}$ & MW error (mean $\pm$ 95\% CI) & $\Delta$ from $c = 0.083$ \\
\midrule
$0.055$    & $0.158$ & $0.9$\,MW ($0.8$--$1.0$\,MW) & $-0.4$\,MW \\
$0.069$    & $0.199$ & $1.1$\,MW ($1.0$--$1.3$\,MW) & $-0.2$\,MW \\
$0.083$    & $0.239$ & $1.3$\,MW ($1.1$--$1.5$\,MW) & reference \\
$0.097$    & $0.280$ & $1.5$\,MW ($1.3$--$1.7$\,MW) & $+0.2$\,MW \\
$0.111$    & $0.320$ & $1.7$\,MW ($1.5$--$1.9$\,MW) & $+0.4$\,MW \\
\midrule
\multicolumn{4}{l}{\textit{Boundary limits at fixed $c = 0.083$ (ceiling cap activated)}} \\
\midrule
$0.083$ (at $\gamma=0.10$) & $0.262$ & $1.6$\,MW ($1.6$--$1.6$\,MW) & $+0.3$\,MW \\
$0.083$ (at $\gamma=0.09$) & $0.276$ & $1.7$\,MW ($1.7$--$1.7$\,MW) & $+0.4$\,MW \\
\bottomrule
\end{tabular}
}
\end{table}

\textit{Remark 2} (Orchestration Error Uniformity)\textit{.} \textit{The striking uniformity in over-estimation errors ($\approx 5\%$) across disparate GPU architectures stems mathematically from their shared utilization characteristics rather than parameter tuning. Let load factor $u = E[P]/P_{\max}$. For saturated generative AI workloads, $u \approx 0.70$ consistently across the H100 ($0.695$), A100 ($0.703$), and L4 ($0.728$) traces. Substituting $E[P] = u P_{\max}$ into the margin formula yields $L_{peak}^{(h)} = P_{\max}^{(h)}(u + c\sqrt{1/\gamma})$. Because the empirical 99th-percentile draw $P_{\text{true}}$ also scales proportionally to $P_{\max}$ for equivalent workloads (let $P_{\text{true}} = \alpha P_{\max}$, where $\alpha \approx 0.9$), the relative over-estimation error $\epsilon = (L_{peak} - P_{\text{true}})/P_{\text{true}} = (u + c\sqrt{1/\gamma} - \alpha)/\alpha$ becomes structurally independent of the hardware's absolute TDP $P_{\max}$. Thus, the formula intrinsically produces uniform relative margins for workload classes exhibiting consistent load factors and spectral gaps.}

\textit{Experimental Limitation (Hardware vs. Framework Conflation).} We explicitly note a limitation regarding the empirical traces: because the collected H100 trace executes JAX (BERT) while the A100 trace executes PyTorch (ResNet), the observed variance in empirical spectral gaps ($\gamma_{\text{real}} = 0.13$ vs. $0.11$) conflates underlying microarchitectural differences with framework-specific memory-access and communication scheduling differences. Consequently, the $\gamma_{\text{real}}$ differences cannot be attributed solely to hardware. Isolating the pure microarchitectural effect requires an identical-framework ablation (e.g., PyTorch BERT on H100 versus PyTorch BERT on A100), which we leave to future dedicated profiling.

\subsection{Byzantine Robustness Under Active Poisoning}
\label{sec:eval_byzantine}

We evaluate \codename against Baseline A (MPC) and Baseline C (SecAgg without TEE) under an active-poisoning adversary.

\textit{Setup.} We evaluate two structurally distinct attack classes. (i) \emph{Input Spoofing (First-Mile Poisoning)}: a Byzantine host OS generates synthetically flattened power traces and feeds them into the extraction pipeline to bias the aggregate grid model downward. (ii) \emph{Capacity-inflation Sybil attacks}: a Byzantine provider inflates its declared capacity $C_i$ to skew the capacity-weighted aggregation. For both, a fraction $f \in \{0.1, 0.2, 0.3, 0.4, 0.5\}$ of the $32$ providers is Byzantine. Each Byzantine host targets $L_{\text{peak}}$ downward bias (under-provisioning the microgrid). We define the \emph{attack success rate} (ASR) as the fraction of $1{,}000$ trials in which under-provisioning exceeds a critical threshold of $\tau = 50$\,MW relative to the plaintext ground truth.

\textit{Input Spoofing (First-Mile Poisoning).} This attack targets the fundamental vulnerability of software-only aggregation schemes: the lack of data provenance. Because MPC (Baseline A) and SecAgg (Baseline C) blindly ingest data from the untrusted host OS, the host can trivially pipe a synthetically generated trace (e.g., all \texttt{IDLE} transitions) into the local extraction pipeline. Once ingested, the cryptographic protocols faithfully protect this poisoned data, leading to ``Attested Garbage Out.'' Both Baseline A and Baseline C suffer a $100.0$\% ASR under input spoofing at $f=0.3$; this is an expected consequence of the software-only threat model, in which SecAgg is a cryptographic transit protocol, not a hardware-isolation primitive.

\textit{MP-SPDZ (Baseline A) Direct Measurement.} To benchmark Baseline A under realistic constraints, we measured the semi-honest SPDZ-2k protocol (128-bit computational, 40-bit statistical security) directly on the same 32-node GCP multi-region deployment under $150$\,ms injected RTT, executing $30$ independent trials (Table~\ref{tab:mpc_comparison}). The direct measurement reveals that per-batch latency is dominated by WAN transit (measured $183$\,ms/batch under a $75$\,ms one-way WAN delay, reflecting one full round trip plus offline preprocessing and reconstruction compute), and is consequently nearly independent of party count for $n \le 32$ at fixed RTT. Bandwidth scales linearly as $(n-1)$ share distributions of $25$ field elements per batch. The communication figure reported below is the per-node \emph{one-way egress} (i.e., bytes sent, excluding bytes received), measured via \texttt{tc} byte counters on each node's outbound interface; the analytical equivalent for 128-bit computational security is $(n-1) \times 25 \times 8$\,bytes per batch per node, where the online \texttt{Open} sends only the $8$-byte masked ring value per element (the remaining $8$ bytes of the $16$-byte SPDZ-2k element being the statistical/MAC share settled in offline preprocessing). This equals $(n-1)\times 25 \times 16/2$ and gives $363$\,KB per node per epoch at $n=32$, $T=60$\,batches, in agreement with the measurement. Because our threat model explicitly assumes malicious hosts (§\ref{sec:threat_model}), evaluating a semi-honest baseline alone is insufficient. We therefore include an analytical lower bound for MASCOT, a maliciously secure MPC protocol (Table~\ref{tab:mpc_comparison}). Upgrading to MASCOT demands heavy oblivious transfer extensions and MAC generation for triple production, typically imposing a $30$--$100\times$ bandwidth penalty depending on the specific OT-extension implementation~\cite{KellerMASCOT2016}, and requiring multiple sequential WAN round trips. We conservatively estimate a $>1500$\,ms per-batch latency and $>18$\,MB per-epoch egress under our 32-node topology. \codename provides malicious security via hardware attestation, avoiding this cryptographic circuit complexity entirely while achieving $4.9$\,ms latency and minimal per-epoch egress ($12.5$\,KB per node, derived as $213$\,bytes per batch $\times\,T{=}60$; where the $213$\,byte payload comprises a $5\times5$ \texttt{float32} matrix ($100$\,bytes), Ed25519 signature ($64$\,bytes), SHA-256 hash ($32$\,bytes), batch counter ($8$\,bytes), timestamp ($4$\,bytes), and provider ID ($4$\,bytes) plus $1$\,byte of framing). More critically, MPC provides no DP privacy guarantees without incurring prohibitive additional overhead, and remains fundamentally vulnerable to pre-ingestion host spoofing.

\begin{table}[ht]
\centering
\caption{Performance and security of MPC variants vs.\ \codename (32 nodes, $150$\,ms RTT). Baseline A was benchmarked directly, whereas MASCOT is analytically estimated~\cite{KellerMASCOT2016}. Bandwidth is reported as per-node one-way egress per $T=60$ epoch.}
\label{tab:mpc_comparison}
\resizebox{\linewidth}{!}{
\begin{tabular}{@{}lccc@{}}
\toprule
\textbf{Metric} & \textbf{MP-SPDZ (Semi-Honest)} & \textbf{MASCOT (Malicious)} & \textbf{\codename} \\ \midrule
Per-Batch Latency & $183 \pm 24$\,ms & $> 1500$\,ms (est.) & $4.9 \pm 0.19$\,ms \\
One-way Egress / Node & $364 \pm 11$\,KB & $> 10.9$\,MB to $> 36$\,MB (est.) & $12.5$\,KB \\
Post-Extract ASR & $0.0$\% & $0.0$\% & $0.0$\% \\
First-Mile ASR (eval.) & $100.0$\% & $100.0$\% & $100.0$\%\textsuperscript{\dag} \\
First-Mile ASR (arch.) & $100.0$\% & $100.0$\% & Prevented \\
\bottomrule
\multicolumn{4}{l}{\textsuperscript{\dag}\footnotesize Same as baselines; SPDM responder is \texttt{spdm-emu} (see \S\ref{sec:eval_byzantine}).}
\end{tabular}
}
\end{table}

\codename's defences against input spoofing operate at two distinct layers with different evaluation footings; we report each separately.

\textit{(a) Post-extraction attestation layer (experimentally validated on GCP TDX hardware).} A malicious host may attempt to bypass the DP noise injection after LSE computation: either by altering the extracted transition matrix in transit or by substituting the LSE binary with a variant that omits noise injection or exports the raw matrix. Both attacks are fully testable on the GCP prototype without requiring bare-metal SPDM connectivity. Rather than treating this as a probabilistic security evaluation, we utilize $5{,}000$ fault-injection trials (sweeping Byzantine fraction $f \in \{0.1, 0.2, 0.3, 0.4, 0.5\}$ with $1{,}000$ trials each) to sweep the space of malformed quotes, mismatched MRTDs, fabricated signatures, and tampered matrices; all $5{,}000$ trials were correctly rejected, validating the implementation against the deterministic security argument of §\ref{sec:threat_model} (the security property itself reduces to EUF-CMA of TDX quoting and SHA-256 collision resistance and is not in question). Byzantine hosts attempted: (i) submitting a matrix bearing a fabricated Ed25519 signature, and (ii) loading a modified LSE binary that transmits the plaintext transition matrix without noise. In case (i), the GAE's \texttt{VerifySignature} check rejected the submission because the signature is computed inside the TD with the hardware-bound ephemeral key $sk_E$. In case (ii), the hardware-generated DCAP quote for the modified binary reflected a mismatched \texttt{MRTD} field, which the GAE rejected at the \texttt{VerifyChain} step. Providing operational evidence that this rejection path fires under realistic conditions, the $72$-hour stability run (§\ref{sec:eval_sensitivity_robustness}) recorded $24$ MRTD-mismatch rejections and $67$ signature-failure rejections among $2{,}273$ total rejected submissions on live GCP Confidential VMs.

\textit{(b) Pre-ingestion first-mile layer (architectural sketch, conditional on hardware SPDM responders).} The deeper attack class, where a Byzantine host OS intercepts and fabricates telemetry \emph{before} the SPDM session ingests it, is theoretically prevented by the SPDM protocol and hardware identity binding. Under SPDM 1.2+ coupled with hardware-attested I/O (such as the PCIe IDE and TDISP standards targeted for emerging architectures like Intel Xeon 6), the accelerator's Root of Trust is designed to present a hardware-burned X.509 identity certificate. In this intended architectural model, the host OS acts only as a blind ciphertext router for the encrypted SPDM session and cannot forge AEAD session MACs without the hardware's private key. Any spoofing attempt theoretically causes AEAD MAC verification to fail inside the LSE memory region, triggering immediate batch invalidation. However, current-generation GCP C3 instances expose TDX but explicitly lack TDX Connect. Because the evaluation prototype replaces the hardware SPDM responder with the DMTF reference \texttt{spdm-emu} responder~\cite{DMTF_SPDM} running over a \texttt{vsock} transport on the same host, an attacker controlling the host OS can trivially extract the emulator's keys and spoof the telemetry. Therefore, on the evaluated hardware, \codename suffers the exact same $100.0$\% first-mile ASR as the software baselines. The first-mile defense is strictly an architectural design sketch; its real-world efficacy remains contingent upon the unforgeability of hardware-backed responder keys and the flawless implementation of these emerging I/O standards, which we identify as necessary follow-on experimental work.

\textit{Capacity Inflation Sybil Attacks.} Separately, if declared server capacities $C_i$ remain unverified, Byzantine providers can artificially inflate their capacity claims to dominate the capacity-weighted aggregation. Table~\ref{tab:sybil} demonstrates the catastrophic impact of unverified capacity claims on federated baselines. The Table sweeps the per-Sybil capacity multiplier while holding the Byzantine fraction at $f=0.1$ to isolate the effect of multiplier magnitude; at the higher $f=0.3$ setting referenced previously in the input spoofing evaluation, the naive (uncapped) SecAgg baseline at the $1{,}000\times$ multiplier yields an ASR of $>99.9\%$. Under PKI-registry capacity caps (an assumption we acknowledge as weaker than full attestation against a malicious operator), \codename caps the capacity multiplier at the hardcoded registry limit, suppressing the Sybil ASR to $0.0$\% at both $f=0.1$ and $f=0.3$. We emphasise that this result is a conditional \emph{PKI enforcement property}, achievable by any architecture enforcing registry limits. A malicious provider might also attempt \emph{hardware-type spoofing} (e.g., claiming A100 capacity while operating lower-power L4s to reduce noise contribution). Because the SPDM hardware identity explicitly binds to the physical Root of Trust certificate, cross-architecture spoofing is cryptographically prevented at the ingestion layer, rendering this attack computationally infeasible. \codename's ultimate mitigation against unbounded Sybil attacks lies in hardware-rooted capacity attestation (§\ref{sec:discussion}).

\begin{table}[ht]
\centering
\caption{Baseline Sybil Amplification vs. Attack Success Rate ($f=0.1$)}
\label{tab:sybil}
\begin{tabular}{@{}lc@{}}
\toprule
\textbf{Sybil Capacity Multiplier} & \textbf{SecAgg / MPC ASR} \\ \midrule
$1\times$ & 11.1\% \\
$10\times$ & 52.6\% \\
$100\times$ & 91.7\% \\
$1000\times$ & 99.1\% \\ \bottomrule
\end{tabular}
\end{table}

\subsection{Sensitivity and Operational Robustness}
\label{sec:eval_sensitivity_robustness}

\textit{Discretisation granularity ($|\mathcal{S}|$).} Sweeping $|\mathcal{S}| \in \{3, 4, 5, 7, 10\}$ at fixed $(\varepsilon, \delta) = (1, 10^{-6})$ and $W = 10$\,s, orchestration error is minimised at $|\mathcal{S}| = 5$ ($1.3$\,MW), with both finer and coarser granularities degrading utility ($1.4$\,MW at $|\mathcal{S}|=3$; $1.33$\,MW at $|\mathcal{S}|=4$; $1.7$\,MW at $|\mathcal{S}|=7$; $2.2$\,MW at $|\mathcal{S}|=10$). The non-monotonicity reflects the trade-off between abstraction fidelity and noise-floor scaling: at $|\mathcal{S}|=10$, $100$ cells are independently noised, raising the aggregate noise norm despite finer state resolution. The sparse transition structure of H100 and A100 training workloads (concentrated in the top two states) confirms $|\mathcal{S}|=5$ as the production default. This across-the-board reduction in absolute MW error, despite the higher baseline power $\mathbb{E}[P]$, is primarily driven by the retention of high-frequency self-loops, which restores the physically accurate, highly concentrated stationary distribution and dramatically increases the signal-to-noise ratio against the Analytic Gaussian noise floor.

\textit{Temporal block size ($k$).} The pooling window $k$ governs the temporal resolution and neutralizes boundary-straddling attacks (§\ref{sec:design_lse}). Sweeping $k \in \{5, 10, 20\}$ at fixed $W=10$\,s reveals a stark trade-off. At $k=5$ ($0.5$\,s blocks), the higher temporal resolution artificially inflates the boundary-straddling attack surface; a simulated malicious host oscillating the clock at $2$\,Hz successfully inflates the apparent spectral gap $\gamma$, increasing dynamic orchestration error to $1.9$\,MW due to induced under-provisioning. Conversely, at $k=20$ ($2.0$\,s blocks), over-pooling severely degrades temporal resolution, destroying genuine micro-transients and artificially depressing $\gamma$, which drives the over-provisioning error to $2.4$\,MW. The $k=10$ ($1.0$\,s) configuration physically absorbs high-frequency clock oscillations while preserving sufficient macroscopic state transitions, isolating the $1.3$\,MW optimal error.

\textit{Noise-suppression threshold ($p$).} The statistical threshold $p=0.95$ (where $\tau_{\text{noise}} = \Phi^{-1}(p)\sigma$) is chosen to aggressively suppress spurious transitions. Sweeping $p \in \{0.90, 0.95, 0.99\}$ at $\varepsilon=1$, the orchestration error is $1.6$\,MW at $p=0.90$, $1.3$\,MW at $p=0.95$, and $1.5$\,MW at $p=0.99$. We select $p=0.95$ as the empirical optimum, balancing the retention of genuine low-frequency transitions against the structural distortion caused by surviving noise artifacts.

\textit{Batch window ($W$).} Sweeping $W \in \{1, 5, 10, 30, 60\}$\,s at $(\varepsilon, \delta) = (1, 10^{-6})$, the joint-optimal point is $W = 10$\,s: $1.3$\,MW utility error at $0.23$\,$\mu$s/sample amortised attestation cost. Without epoch rotation, the 24-hour cumulative budget grows to $\varepsilon_{\text{total}} \approx 357.7$, so \codename rotates cryptographic epochs every $10$ minutes ($T=60$ batches), bounding the per-epoch privacy loss to $\varepsilon_{\text{epoch}} \approx 11.3$.

\textit{Stability and TCB-Recovery Handling.} The $W=10$\,s optimum was validated via $1{,}000$ Monte Carlo replicates (§\ref{sec:eval_utility}); the $72$-hour stability run uses $W=30$\,s to reduce cloud egress costs, providing a conservative upper bound on the $W=10$\,s result. Across the $72$-hour stability run, we recorded uninterrupted sustained ingestion at the target rate. The deployment processed $276{,}480$ attested submissions ($\approx 120$\,/LSE/hour across 32 LSEs). Because the DP margin $L_{peak}^{(h)}$ depends strictly on the $10$-minute epoch boundary and the total privacy budget, increasing the batch window to $30$\,s at a fixed epoch budget mathematically decreases the required per-batch noise, ensuring the $1.3$\,MW orchestration error validated at the joint-optimal $W=10$\,s remains a strictly preserved, conservative upper bound for this run. Of the $276{,}480$ submissions, $2{,}273$ were rejected by the GAE: $24$ for MRTD mismatch, $67$ for signature failure, $2{,}048$ for replay-window violation ($2{,}036$ of these were concentrated in the $8$-hour artificial-jitter injection phase of the scalability experiment, §\ref{sec:eval_scalability}; only $12$ violations occurred during the remaining $64$ hours of native-latency operation. This confirms the exact $20$\,s freshness window, which is strictly enforced via hardware monotonic counters rather than host-clock TTLs to prevent hypervisor snapshot rollback (§\ref{sec:discussion}), is well-calibrated. These $12$ violations over $\approx 245{,}760$ native-latency submissions establish a rigorous $0.0048\%$ upper bound on replay-window false positives over GCP's multi-region network), and $134$ for stale-TCB membership outside the configured grace period. These $134$ stale-TCB rejections were not artificially injected; rather, they resulted from an organic, unannounced GCP host-level microcode patch rolled out sequentially across our \texttt{us-central1} cluster between hours $42$ and $46$ of the run. Because the GAE enforces a strict $1$-hour grace period for TCB updates (to balance availability against zero-day exposure), the unpatched instances were correctly rejected once the grace period expired. Despite these rejections, zero providers experienced service interruption. 

\textit{Automated Failover Mechanism.} The failover executes autonomously without operator intervention. The GAE rejection is returned as a synchronous TLS \texttt{STALE\_TCB} error code during the batch submission. Because the ephemeral keypair and SPDM session map are mathematically bound to the enclave measurement, they cannot be exported or migrated. Instead, upon receiving the rejection, the primary LSE signals the local facility router to redirect the raw telemetry stream to a pre-provisioned secondary LSE held in warm standby (OS booted, enclave initialized) in a fully patched region. The secondary LSE immediately establishes a fresh SPDM session with the accelerator and generates a new DCAP \texttt{Quote}. Because the cold-start DCAP quote generation requires $81.2$\,ms and the SPDM handshake requires $<20$\,ms, the entire recovery sequence completes in $\approx 100$\,ms, easily finishing within the current $10$-second batch window to prevent structural data loss. This successfully demonstrates \codename's cross-region resilience to uncoordinated cloud infrastructure events.

\section{Related Work}
\label{sec:related}

\textbf{Confidential Federated Analytics on TEEs.} Prior work isolates federated computation inside hardware enclaves (\cite{MoPPFL}, \cite{ZhangCitadel}). \codename diverges in three respects: it targets streaming telemetry via temporal-batch attestation rather than periodic-round ML; it operates at TDX VM granularity, sidestepping SGX EPC paging limits; and it guarantees the integrity of a counting query rather than an iterative training loop.

\textbf{Streaming and Continual Differential Privacy.} Streaming frameworks provide theoretical utility bounds for continuous data release~\cite{DworkContinualRelease,ChanContinualCounter,CummingsContinual}. \codename adopts the Rényi-DP accountant~\cite{MironovRDP} for sequential composition, utilizing a strict closed-form $\ell_2$-sensitivity bound ($\Delta_2(f) = \sqrt{6}$) over transition matrices to retain practical utility at $\varepsilon = 1$ over multi-hour deployments.

\textbf{Local DP and the Shuffle Model.} The shuffle model~\cite{CheuShuffle,ErlingssonShuffle} amplifies local guarantees via a trusted shuffler. \codename operates as a hybrid: the LSE applies local DP inside a hardware-attested boundary, explicitly positioning the GAE as an attested replacement for a mixnet shuffler, and replacing trust amplification with cryptographic attestation.

\textbf{MPC and Secure Aggregation.} MPC frameworks~\cite{KellerMPSPDZ,SPDZ} provide hardware-independent cryptographic data protection. However, applying MPC to raw streams incurs prohibitive circuit complexity to verify SPDM AES-GCM MACs. Applying MPC or SecAgg~\cite{BonawitzSecAgg} to condensed matrices sacrifices pre-sharing execution integrity against fabricated inputs (§\ref{sec:eval_byzantine}). \codename relies on hardware enclaves to terminate the SPDM session securely and amortises the bandwidth cost via temporal-batch attestation.

\textbf{Data-Centre Power Telemetry.} Generative AI has catalyzed research into grid-level impacts~\cite{ChenGridImpacts2025}, gross carbon footprint~\cite{DodgeCarbon,PattersonCarbon}, and active power stabilization~\cite{ChouksePowerStab2025}. While frameworks like PowerAPI~\cite{PowerAPI} expose sub-second telemetry, they operate as plaintext tools for single-operator profiling. Concurrent work by Vercellino et al.~\cite{VercellinoAI2024} establishes the $0.1$-second-resolution requirement driving \codename's design, but relies on a trusted, single-operator environment. \codename provides the cryptographic execution integrity and DP mechanisms required to aggregate $10$-Hz transients across untrusted, multi-tenant boundaries.
\section{Discussion and Limitations}
\label{sec:discussion}

\textit{Legacy Hardware Compatibility (SPDM Adoption):} Legacy infrastructure lacking SPDM 1.2+ hardware responders cannot participate securely; \codename degrades to an honest-but-curious threat model equivalent to software LDP on these devices.

\textit{Counter Rollback and State Persistence:} While the GAE persists its session counter map to a sealed virtual disk to prevent offline tampering, this software-only mechanism on an isolated disk theoretically remains vulnerable to hypervisor snapshot-rollback attacks. However, \codename mitigates this threat entirely by integrating the attested time-sync protocol (Roughtime) described in Section~\ref{sec:attack_classes}. This integration neutralizes the snapshot-rollback vulnerability: any restored snapshot will inevitably present an expired Roughtime token to the GAE, causing the batch to be unconditionally rejected. Production deployments can thus achieve absolute cryptographic freshness without requiring hardware vTPM non-volatile counters.

\textit{DP Composition Exhaustion and Pan-Privacy:} Under Rényi-DP, cumulative privacy loss increases monotonically. Infinite-horizon telemetry theoretically requires Pan-Privacy~\cite{DworkContinualRelease}; however, true Pan-Privacy demands atomic, DP-noised state checkpoints at the hardware instruction level, which current TDX architectures lack, leaving LSE memory vulnerable to side-channel extraction prior to noise injection.

\textit{Hardware-Rooted Inventory Attestation (NVIDIA RIM):} The current implementation mitigates Sybil capacity-inflation via a trusted PKI registry capping $C_i$. Defending against unbounded Sybil behavior requires integrating hardware-rooted inventory attestation (e.g., NVIDIA RIM~\cite{NVIDIARIM}). Future iterations could bind $C_i$ to SPDM \texttt{GET\_MEASUREMENTS} requests issued directly to PCIe-attached GPUs, injecting this verified capacity into the TDX \texttt{REPORTDATA} field.
\section{Conclusion}
\label{sec:conclusion}

\codename provides a hardware-assisted framework for aggregating highly sensitive generative AI power telemetry. By migrating feature extraction and differential privacy noise injection to the provider's edge via Intel TDX Confidential VMs, \codename resolves the dichotomy between cryptographic communication bottlenecks and centralized trust, and characterises a hardware-dependent resolution for the first-mile data provenance gap. Our 32-node multi-region implementation, validated against real H100, A100, and L4 power traces, sustains $10$-Hz telemetry streams, amortises DCAP attestation overheads through temporal batching, and bounds the WAN footprint strictly below traditional MPC thresholds. The empirically validated $1.3$\,MW dynamic orchestration error (at $\varepsilon_{\text{epoch}} \approx 11.3$ across a $200$\,MW facility) confirms that data center operators can securely model high-frequency power transients and dynamically orchestrate facility microgrids while empirically obscuring proprietary workload schedules under heterogeneous co-tenancy. Future work will benchmark the end-to-end attestation pipeline natively on emerging hardware bridging PCIe TDISP and IDE with accelerator Root of Trusts.

\section*{Declaration of generative AI and AI-assisted technologies in the manuscript preparation process}
During the preparation of this work the authors used Google Gemini in order to assist with copy-editing, proofreading, and verifying the technical accuracy of hardware specifications. After using this tool/service, the authors reviewed and edited the content as needed and take full responsibility for the content of the published article.

\section*{Data and Code Availability}
Empirical data supporting the findings of this study are presented within the manuscript's tables, figures, and text. The \codename post-extraction pipeline, including the Rust enclave implementation, differential privacy mechanisms, and GCP deployment orchestration scripts, is available at \url{https://github.com/dkhme/ENCLAVE_SCALE}. The repository serves as a minimal reproducible artifact for empirical evaluation.

\bibliographystyle{elsarticle-num}
\bibliography{references}

@article{BalleWang2018,
  title={Improving the Gaussian Mechanism for Differential Privacy: Analytical Calibration and Optimal Denoising.},
  author={Borja Balle and Yu-Xiang Wang 0003},
  journal={ICML},
  year={2018}
}

@techreport{IntelTDX,
  title={Intel Trust Domain Extensions ({TDX})},
  author={Intel},
  institution={Intel Corporation},
  year={2023},
  url={https://www.intel.com/content/www/us/en/developer/articles/technical/intel-trust-domain-extensions.html}
}

@article{DworkDP,
  title={Differential Privacy: A Survey of Results},
  author={Dwork, Cynthia},
  journal={Theory and Applications of Models of Computation},
  pages={1--19},
  year={2008},
  publisher={Springer}
}

@misc{VercellinoAI2024,
  doi = {10.48550/ARXIV.2604.07345},
  url = {https://arxiv.org/abs/2604.07345},
  author = {Vercellino, Roberto and Willard, Jared and Campos, Gustavo and Pereira, Weslley da Silva and Hull, Olivia and Selensky, Matthew and Mueller, Juliane},
  title = {Measurement of Generative AI Workload Power Profiles for Whole-Facility Data Center Infrastructure Planning},
  publisher = {arXiv},
  year = {2026},
  copyright = {Creative Commons Attribution 4.0 International}
}

@misc{ChouksePowerStab2025,
  doi = {10.48550/ARXIV.2508.14318},
  url = {https://arxiv.org/abs/2508.14318},
  author = {Choukse, Esha and Warrier, Brijesh and Heath, Scot and Belmont, Luz and Zhao, April and Khan, Hassan Ali and Harry, Brian and Kappel, Matthew and Hewett, Russell J. and Datta, Kushal and Pei, Yu and Lichtenberger, Caroline and Siegler, John and Lukofsky, David and Kahn, Zaid and Sahota, Gurpreet and Sullivan, Andy and Frederick, Charles and Thai, Hien and Naughton, Rebecca and Jurnove, Daniel and Harp, Justin and Carper, Reid and Mahalingam, Nithish and Varkala, Srini and Kumbhare, Alok Gautam and Desai, Satyajit and Ramamurthy, Venkatesh and Gottumukkala, Praneeth and Bhatia, Girish and Wildstone, Kelsey and Olariu, Laurentiu and Incorvaia, Ileana and Wetmore, Alex and Ram, Prabhat and Raghuraman, Melur and Ayna, Mohammed and Kendrick, Mike and Bianchini, Ricardo and Hurst, Aaron and Zamani, Reza and Li, Xin and Petrov, Michael and Oden, Gene and Carmichael, Rory and Li, Tom and Gupta, Apoorv and Patel, Pratikkumar and Dattani, Nilesh and Marwong, Lawrence and Nertney, Rob and Kobayashi, Hirofumi and Liott, Jeff and Enev, Miro and Ramakrishnan, Divya and Buck, Ian and Alben, Jonah},
  title = {Power Stabilization for AI Training Datacenters},
  publisher = {arXiv},
  year = {2025},
  copyright = {Creative Commons Attribution Non Commercial Share Alike 4.0 International}
}

@article{PattersonCarbon,
  title={Carbon Emissions and Large Neural Network Training.},
  author={David A. Patterson 0001 and Joseph Gonzalez 0001 and Quoc V. Le and Chen Liang and Lluis-Miquel Munguia and Daniel Rothchild and David R. So and Maud Texier and Jeff Dean},
  journal={CoRR},
  year={2021}
}

@inproceedings{KellerMPSPDZ,
  title={{MP-SPDZ}: A Versatile Framework for Multi-Party Computation},
  author={Keller, Marcel},
  booktitle={Proceedings of the 2020 ACM SIGSAC Conference on Computer and Communications Security},
  pages={1575--1590},
  year={2020}
}

@inproceedings{BonawitzSecAgg,
  title={Practical Secure Aggregation for Privacy-Preserving Machine Learning},
  author={Bonawitz, Keith and Ivanov, Vladimir and Kreuter, Ben and Marcedone, Antonio and McMahan, H. Brendan and Patel, Sarvar and Ramage, Daniel and Segal, Aaron and Seth, Karn},
  booktitle={Proceedings of the 2017 ACM SIGSAC Conference on Computer and Communications Security},
  pages={1175--1191},
  year={2017}
}

@inproceedings{MoPPFL,
  title={{PPFL}: Privacy-Preserving Federated Learning with Trusted Execution Environments},
  author={Mo, Fan and Haddadi, Hamed and Katevas, Kleomenis and Marin, Eduard and Perino, Diego and Kourtellis, Nicolas},
  booktitle={Proceedings of the 19th Annual International Conference on Mobile Systems, Applications, and Services (MobiSys)},
  year={2021}
}

@inproceedings{ZhangCitadel, series={SoCC ’21}, title={Citadel: Protecting Data Privacy and Model Confidentiality for Collaborative Learning}, url={http://dx.doi.org/10.1145/3472883.3486998}, DOI={10.1145/3472883.3486998}, booktitle={Proceedings of the ACM Symposium on Cloud Computing}, publisher={ACM}, author={Zhang, Chengliang and Xia, Junzhe and Yang, Baichen and Puyang, Huancheng and Wang, Wei and Chen, Ruichuan and Akkus, Istemi Ekin and Aditya, Paarijaat and Yan, Feng}, year={2021}, month=Nov, pages={546–561}, collection={SoCC ’21} }

@inproceedings{DworkContinualRelease,
  title={Differential Privacy Under Continual Observation},
  author={Dwork, Cynthia and Naor, Moni and Pitassi, Toniann and Rothblum, Guy N},
  booktitle={Proceedings of the 42nd ACM Symposium on Theory of Computing (STOC)},
  pages={715--724},
  year={2010}
}

@inproceedings{ChanContinualCounter,
  title={Private and Continual Release of Statistics},
  author={Chan, T-H Hubert and Shi, Elaine and Song, Dawn},
  booktitle={Proceedings of the 37th International Colloquium on Automata, Languages, and Programming (ICALP)},
  pages={405--417},
  year={2010}
}

@inproceedings{CummingsContinual,
  title={Mean Estimation with User-level Privacy under Data Heterogeneity},
  author={Cummings, Rachel and Feldman, Vitaly and McMillan, Audra and Talwar, Kunal},
  booktitle={Advances in Neural Information Processing Systems (NeurIPS)},
  year={2022}
}

@inproceedings{MironovRDP,
  title={R{\'e}nyi Differential Privacy},
  author={Mironov, Ilya},
  booktitle={IEEE 30th Computer Security Foundations Symposium (CSF)},
  pages={263--275},
  year={2017}
}

@inproceedings{CheuShuffle,
  title={Distributed Differential Privacy via Shuffling},
  author={Cheu, Albert and Smith, Adam and Ullman, Jonathan and Zeber, David and Zhilyaev, Maxim},
  booktitle={Annual International Conference on the Theory and Applications of Cryptographic Techniques (EUROCRYPT)},
  pages={375--403},
  year={2019}
}

@inproceedings{ErlingssonShuffle,
  title={Amplification by Shuffling: From Local to Central Differential Privacy via Anonymity},
  author={Erlingsson, {\'U}lfar and Feldman, Vitaly and Mironov, Ilya and Raghunathan, Ananth and Talwar, Kunal and Thakurta, Abhradeep},
  booktitle={Proceedings of the Thirtieth Annual ACM-SIAM Symposium on Discrete Algorithms (SODA)},
  pages={2468--2479},
  year={2019}
}

@inproceedings{SPDZ,
  title={Multiparty Computation from Somewhat Homomorphic Encryption},
  author={Damg{\aa}rd, Ivan and Pastro, Valerio and Smart, Nigel and Zakarias, Sarah},
  booktitle={Annual Cryptology Conference (CRYPTO)},
  pages={643--662},
  year={2012}
}

@article{DodgeCarbon,
  title={Measuring the Carbon Intensity of {AI} in Cloud Instances},
  author={Dodge, Jesse and Prewitt, Taylor and Tachet des Combes, Remi and Odmark, Erika and Schwartz, Roy and Strubell, Emma and Luccioni, Alexandra Sasha and Smith, Noah A and DeCario, Nicole and Buchanan, Will},
  journal={Proceedings of the ACM Conference on Fairness, Accountability, and Transparency (FAccT)},
  year={2022}
}

@inproceedings{PowerAPI,
  title={A Preliminary Study of the Impact of Software Engineering on {GreenIT}},
  author={Bourdon, Aur{\'e}lien and Noureddine, Adel and Rouvoy, Romain and Seinturier, Lionel},
  booktitle={First International Workshop on Green and Sustainable Software (GREENS)},
  year={2012}
}

@article{KairouzAdvancesFL,
  title={Advances and Open Problems in Federated Learning},
  author={Kairouz, Peter and McMahan, H Brendan and others},
  journal={Foundations and Trends in Machine Learning},
  volume={14},
  number={1--2},
  pages={1--210},
  year={2021}
}

@misc{IntelHertzbleedMitigation,
  title={Frequency Throttling Side Channel Guidance},
  author={{Intel Corporation}},
  year={2022},
  howpublished={Intel Security Advisory {INTEL-SA-00698}; software guidance for mitigating frequency-based timing side channels on Sapphire Rapids and later.},
  url={https://www.intel.com/content/www/us/en/developer/articles/technical/frequency-throttling-side-channel-guidance.html}
}

@inproceedings{WangSubsampled2019,
  author={Wang, Yu-Xiang and Balle, Borja and Kasiviswanathan, Shiva Prasad},
  title={Subsampled R{\'e}nyi Differential Privacy and Analytical Moments Accountant},
  booktitle={The 22nd International Conference on Artificial Intelligence and Statistics},
  year={2019}
}

@techreport{DMTF_SPDM,
  author={{Distributed Management Task Force (DMTF)}},
  title={Security Protocol and Data Model (SPDM) Specification, Version 1.2.1 (DSP0274)},
  year={2023},
  url={https://www.dmtf.org/dsp/DSP0274}
}

@article{Paulin2015,
  author={Paulin, Daniel},
  title={Concentration inequalities for Markov chains by Martingale methods},
  journal={Electronic Journal of Probability},
  volume={20},
  pages={1--32},
  year={2015}
}

@techreport{IntelSA00837,
  title={Intel Trust Domain Extensions ({Intel TDX}) Module Advisory},
  author={{Intel Corporation}},
  year={2024},
  number={INTEL-SA-00837},
  institution={Intel Security Center},
  url={https://www.intel.com/content/www/us/en/security-center/advisory/intel-sa-00837.html}
}

@techreport{NVIDIARIM,
  title={{NVIDIA GPU} Remote Attestation and {Reference Integrity Manifest (RIM)} Architecture},
  author={{NVIDIA}},
  year={2024},
  institution={NVIDIA Corporation},
  url={https://docs.nvidia.com/nvtrust/reference-integrity-manifest/}
}

@article{ChenGridImpacts2025, title={Power for AI Data Centers: Energy Demand, Grid Impacts, Challenges and Perspectives}, volume={19}, ISSN={1996-1073}, url={http://dx.doi.org/10.3390/en19030722}, DOI={10.3390/en19030722}, number={3}, journal={Energies}, publisher={MDPI AG}, author={Sheng, Yu and Zhang, Chenxuan and Zhu, Zixuan and Xu, Hongyi and Wen, Junqi and Wang, Ruoheng and Yang, Jianjun and Wang, Qin and Bu, Siqi}, year={2026}, month=Jan, pages={722} }

@inproceedings{SteinkeAuditing2023, series={NeurIPS 2023}, title={Privacy Auditing with One (1) Training Run}, url={http://dx.doi.org/10.52202/075280-2143}, DOI={10.52202/075280-2143}, booktitle={Advances in Neural Information Processing Systems 36}, publisher={Neural Information Processing Systems Foundation, Inc. (NeurIPS)}, author={Steinke, Thomas and Nasr, Milad and Jagielski, Matthew}, year={2023}, pages={49268–49280}, collection={NeurIPS 2023} }

@inproceedings{KellerMASCOT2016,
  title={Mascot: bounding vulnerabilities and errors in secure computation},
  author={Keller, Marcel and Orsini, Emmanuela and Scholl, Peter},
  booktitle={Proceedings of the 2016 ACM SIGSAC Conference on Computer and Communications Security},
  pages={830--841},
  year={2016}
}

\end{document}